\documentclass[acmsmall, screen]{acmart}

\usepackage{enumitem}
\usepackage{xspace}
\usepackage{tabularx}
\usepackage{tabto}

\newcommand{\revise}[1]{\textcolor{black}{#1}}

\startPage{1}

\usepackage{booktabs}   
\usepackage{subcaption} 
\usepackage{macro}

\setcopyright{rightsretained}
\acmPrice{}
\acmDOI{10.1145/3591248}
\acmYear{2023}
\copyrightyear{2023}
\acmSubmissionID{pldi23main-p163-p}
\acmJournal{PACMPL}
\acmVolume{7}
\acmNumber{PLDI}
\acmArticle{134}
\acmMonth{6}
\received{2022-11-10}
\received[accepted]{2023-03-31}

\begin{document}

\title[ImageEye:  Batch Image Processing using Program Synthesis]{ImageEye:  Batch Image Processing using Program Synthesis}         


\author{Celeste Barnaby}
\orcid{0000-0001-7688-6133}
\affiliation{
  \institution{University of Texas at Austin}
  \country{USA} 
}
\email{celestebarnaby@utexas.edu} 

\author{Qiaochu Chen}
\orcid{0000-0003-4680-5157}             
\affiliation{           
  \institution{University of Texas at Austin}           
  \country{USA}                   
}
\email{qchen@cs.utexas.edu}

\author{Roopsha Samanta}
\orcid{0009-0000-2456-217X}
\affiliation{
    \institution{Purdue University}
    \country{USA}
}
\email{roopsha@purdue.edu}

\author{Işıl Dillig}
\orcid{0000-0001-8006-1230}
\affiliation{
    \institution{University of Texas at Austin}
    \country{USA}
}
\email{isil@cs.utexas.edu}

\begin{abstract}
This paper presents a new synthesis-based approach for batch image processing. Unlike existing  tools that can only apply global edits to the entire image, our method can apply fine-grained edits to individual objects within the image. For example, our method can selectively blur or crop specific objects that have a certain property.  To facilitate such fine-grained image editing tasks, we propose a neuro-symbolic domain-specific language (DSL) that combines pre-trained neural networks for image classification with other language constructs that enable symbolic reasoning. Our method can automatically learn programs in this DSL from user demonstrations by utilizing a novel synthesis algorithm. We have implemented the proposed technique in a tool called \toolname and evaluated it on 50 image editing  tasks. Our evaluation shows that \toolname is able to automate 96\% of these tasks.

\end{abstract}



\maketitle

\section{Introduction}\label{sec:intro}

Because many real-world scenarios require editing a very large number of images, existing photo editing software provides some support for image processing in batch mode. For example, popular software like Adobe Photoshop and Luminar allow users to process multiple files at the same time by specifying  actions like resizing or converting to a specified file type. 

Despite the popularity of such tools, batch editing  capabilities of existing software are extremely limited and can only perform edits globally to the \emph{entire} image. However, many image editing tasks of interest require fine-grained edits to specific parts of the image. For example, consider a scenario where someone wants to upload a collection of their photos after concealing the identities of certain people. Such a task requires performing selective edits (e.g., blurring) to certain parts of the image but not others. As another example, consider a batch processing task where someone wishes to adjust the  white balance of certain types of objects, such as  human faces. Because this requires combining programmatic edits with object classification, existing solutions fall short in successfully automating such image manipulation tasks. 

In this paper, we propose a new technique, based on program synthesis, for automating selective image editing tasks. Given a small set of user demonstrations (performed through a graphical user interface), our  approach automatically synthesizes a program that can be applied to a  much larger set of images. Because these programs are expressed in a neuro-symbolic domain-specific language (DSL) comprised of both logical operators and (pre-trained) neural networks, they can be used to perform \emph{fine-grained  edits} where different actions can be selectively applied to different parts of the image. As a result, our approach can automate image processing tasks that are well beyond the scope of existing tools. 

At a high level, programs in our image editing DSL specify what actions (blur, crop etc.) to apply to what parts of the image. Thus, an image editing program can be viewed as a set of \emph{extractor} and \emph{action} pairs where each extractor selects a \emph{part} of the image and the action specifies what operation to apply to that part. Because these extractors are expressed in a rich vocabulary involving both neural primitives and functional operators, our DSL makes it possible to combine object classification with  relational reasoning. 

Beyond proposing a  DSL for batch image processing, a key contribution of this paper is a new program synthesis technique \revise{for learning programs in this DSL. At a high level, our approach reduces this problem to a more standard programming-by-example (PBE) task by utilizing the concept of \emph{symbolic images}: rather than representing an image as a set of low-level pixels, we represent images as a mapping from object identifiers to their symbolic properties.  This  representation is obtained by applying segmentation to the input image and utilizing neural object classifiers to extract attributes of each detected object. Overall, this symbolic representation is crucial to our technique in two ways: First, it allows defining  a formal semantics of our DSL in terms of sets of high-level objects as opposed to a 2D array of low-level pixels. Second,  due to this symbolic representation, the learning task can be reduced to the problem of  synthesizing an extractor function that  produces a target set of objects from among all objects in the input image.} 

 \revise{
To solve the PBE problem in this context, our approach utilizes  top-down enumerative search, as done in prior work~\cite{lambda2, aws_equiv_reduc, flashextract, scythe}. As standard, the idea is to maintain a worklist of \emph{partial programs} (i.e., programs with unknown parts) that are gradually refined into a concrete implementation.
However, basic enumerative search does not scale to the image manipuation tasks of interest in this work because images often contain \emph{many} objects with \emph{many} different attributes. As a result, the search space becomes enormous,  necessitating novel pruning techniques that can be utilized to rule out redundant or infeasible partial programs. Specifically, our underlying PBE algorithm addresses scalability challenges of this domain through two key insights:}


\begin{enumerate}[leftmargin=*]
    \item {\bf Equivalence reduction with term rewriting and partial evaluation:} Many partial programs enumerated during top-down search are bound to produce the same output image \emph{no matter} how the unknown parts  are instantiated.
    In other words, the basic search procedure ends up enumerating many \emph{redundant} partial programs that can be safely thrown away. To detect such redundancies and prune the search space, our method leverages a combination of \emph{term rewriting} and \emph{partial evaluation} to reason about observational equivalence in the context of images. \revise{While prior work on program synthesis has used term rewriting and partial evaluation \emph{in isolation}, we show that term rewriting is considerably more effective in this context when it is combined with partial evaluation.}
    \item \revise{{\bf Abstract semantics for images:}} Some partial programs enumerated during search can \emph{never} produce the target output image no matter how they are refined into a concrete implementation.  \revise{To avoid such dead-ends in the search space, our method utilizes a novel abstraction (and its corresponding abstract semantics) for image editing programs.  In particular, reasoning backwards from the target image, our method infers the set of objects that \emph{must} and \emph{may} appear in an (unknown) subprogram and uses this information to identify infeasible partial programs in our image editing DSL.} 

\end{enumerate}

We have implemented our proposed approach in a tool called \toolname and evaluated it on 50  image processing tasks inspired by practical tasks and on-line forum discussions. Our evaluation shows that \toolname can successfully automate 48 of these tasks and that it can infer the intended program after a small number of user demonstrations. We also perform  comparisons against simpler synthesis baselines and present ablation studies to demonstrate the usefulness of our proposed synthesis technique. 

To summarize, this paper makes the following key contributions:

\begin{itemize}[leftmargin=*]
    \item We describe the first  solution for automating fine-grained image editing tasks.
    \item We present an image processing DSL that combines neural computer vision primitives with programmatic constructs for  relational reasoning, \revise{and we define the formal semantics of this DSL in terms of the concept of \emph{symbolic images}}. 
    \item We propose a novel synthesis algorithm for generating programs in our DSL from a set of user demonstrations. \revise{Our technique decomposes the overall synthesis problem into a set of independent PBE tasks and utilizes two key ideas (namely, abstraction-guided reasoning about images and combined use of partial evaluation and term rewriting) to allow this approach to scale to realistic image batches.}
    \item We implement our approach in a new tool called \toolname and evaluate it experimentally on dozens of image editing tasks involving a diverse set of images.
\end{itemize}
\section{Overview}\label{sec:overview}


\mypar{Usage scenario} Suppose that a user has a batch of several hundred images from a school recital. The user would like to identify all  images that feature their daughter playing the violin and crop everything else out of those images. Our proposed tool, \toolname, is useful for these types of tasks  that are  easy for a small number of images but grueling for a large batch. 

To automate this task, the user loads their images into \toolname and  identifies a few images that feature their daughter playing a violin. \revise{For each image, they use the \toolname graphical user interface to  select their daughter's face and violin, and choose the ``Crop" action to crop the the selected region. Under the hood, the \toolname GUI uses computer vision models (for image segmentation and object recognition) to highlight the detected objects and allows the user to select each object individually. If the relevant objects   are not detected in the current image, the user will quickly realize that this particular image is not useful for demonstrating their intent and move on to a different image.} 

Once the user has edited a few representative images, they press the “Synthesize” button, which invokes \toolname's synthesis engine and searches for a program $P$ that matches the user's demonstrations. \revise{\toolname then applies $P$ to all images in the batch, and produces a new set of edited images. Next, the user inspects the resulting images to decide whether the synthesized program is correct. If $P$ fails to produce the intended edit for many images in the batch, then the synthesis result is likely incorrect  and the user may provide one or more additional demonstrations as training examples. On the other hand, if the output images look as intended, the synthesized program likely captures the user's intent. However, there may be some imperfections due to shortcomings of the neural models used in the synthesized program. In this case, the user needs to manually edit a small number of images where the resulting image differs from the expected output, but this is still much more convenient than manually editing \emph{all} images in the batch. 
}


\mypar{Shortcomings of existing techniques} While existing image editing tools like Photoshop and GIMP allow some forms of batch processing, they only support simple manipulations (like resizing or applying a filter) to the \emph{entire} image. Notably, such tools do not allow batch processing tasks that differ based on the \emph{content} of each image. \revise{Because our motivating example requires reasoning about the presence of specific objects in each image and applying a cropping action accordingly, such tools are not useful for automating this task.}

On the other hand, tools like Amazon Rekognition use pre-trained neural models for object detection. Given an image, these tools can identify and locate a wide range of objects, including text and human faces. In addition, they can recognize the same face across different images and discern several interesting properties of human faces, such as their approximate age and whether the person is smiling. However, they neither provide functionality for editing individual images nor for batch processing. \revise{Hence, such computer vision models are also not directly useful for performing the task in our motivating example.}

\mypar{Our approach} 
\revise{
In contrast to existing techniques, our approach combines the relative strengths of programmatic batch image processing with  computer vision techniques (image segmentation and object recognition).  Specifically, our approach utilizes a neuro-symbolic DSL that leverages pre-trained neural networks for \emph{perception} and higher-level language constructs for \emph{symbolic reasoning}. The combination of these symbolic and neural constructs is very powerful in that it enables (a) reasoning about relationships between different objects in the image, and (b) selectively editing parts of the image that contain some visual cue of interest.}

In more detail, a program in our neuro-symbolic DSL is of the form $\{E \shortrightarrow A, \dots, E \shortrightarrow A\}$, where $A$ is an action and $E$ is an extractor. An action describes a specific image manipulation (e.g. \textsf{Crop} or \textsf{Blur}),  and an extractor describes the subset of objects to which that action will be applied. As a simple example, consider the program $\{\textsf{Object}(\texttt{cat}) \rightarrow \textsf{Brighten}\}$, which applies a brightening filter to all cats in the input image. 
This DSL also allows combining different objects via  standard set operators. For instance, the extractor $\textsf{Intersect}({\sf Is}(\textsf{Smiling}), \textsf{Complement}({\sf Is}(\textsf{EyesOpen})))$ extracts all human faces  that are smiling and do \emph{not} have their eyes open. In addition, a $\textsf{Find}$ operator can be used to extract objects based on their relative position within the image. For instance, the extractor $\textsf{Find}( {\sf Is}(\textsf{Text}(\texttt{"Total"})), \textsf{TextObject}, \textsf{GetRight})$  first identifies all text objects matching the word \texttt{"Total"} and then, for each such object $o$, it extracts the first text object that is to the right of $o$.

With these basic DSL constructs in mind, let us consider the  program that can be used to automate our target task (i.e., finding and cropping all images that feature the user's daughter with a violin). This task can be expressed with the following program in our DSL:
\begin{align*}
\small
\begin{split}
\{\textsf{Union}(&\textsf{Find}({\sf Is}(\textsf{Face}({\tt Id})), \textsf{Object}(\texttt{violin}), \textsf{GetBelow}), \\
&\textsf{Find}({\sf Is}(\textsf{Object}(\texttt{violin})), \textsf{Face}({\tt Id}), \textsf{GetAbove})) \rightarrow \textsf{Crop}\}
\end{split}
\end{align*}

Here, the extractor is a \textsf{Union} of two sub-extractors: \textsf{Find}({\sf Is}(\textsf{Face}({\tt Id})), \textsf{Object}(\texttt{Violin}), \textsf{GetBelow}) and \textsf{Find}({\sf Is}(\textsf{Object}(\texttt{Violin})), \textsf{Face}({\tt Id}), \textsf{GetAbove}). The first sub-extractor identifies all human faces with the identifier {\tt Id}, where {\tt Id} corresponds to the face of the user’s daughter. Then, for each such face, this program extracts the first violin object located below that face. Conversely, the second sub-extractor identifies all violin objects; then, for each such object, it extracts the first human face with identifier {\tt Id} that is located above the violin. In other words, the first sub-extractor extracts the violin that is played by the user’s daughter, and the second sub-extractor extracts the face of the user’s daughter when she is holding a violin. The union of these sub-extractors precisely describes the part of the image that the user wants to select.

\begin{figure} 
\begin{minipage}[c]{0.6\textwidth}
    \includegraphics[width=0.7\textwidth]{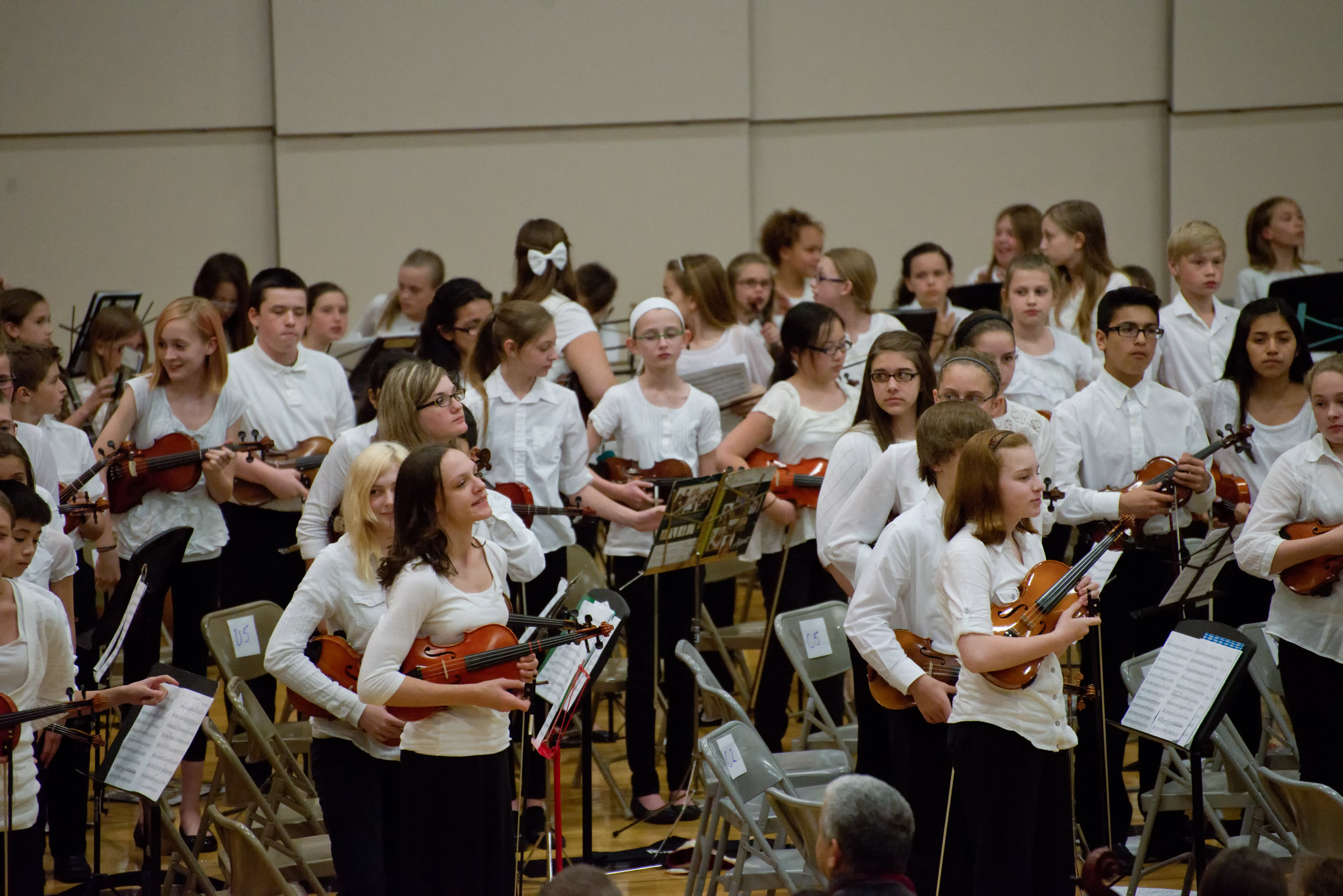}
\end{minipage}
\begin{minipage}[c]{0.05\textwidth}
\centering
\hspace*{-70pt}
$\Longrightarrow$
\end{minipage}
\begin{minipage}[c]{0.24\textwidth}
\small
    \includegraphics[width=0.85\textwidth]{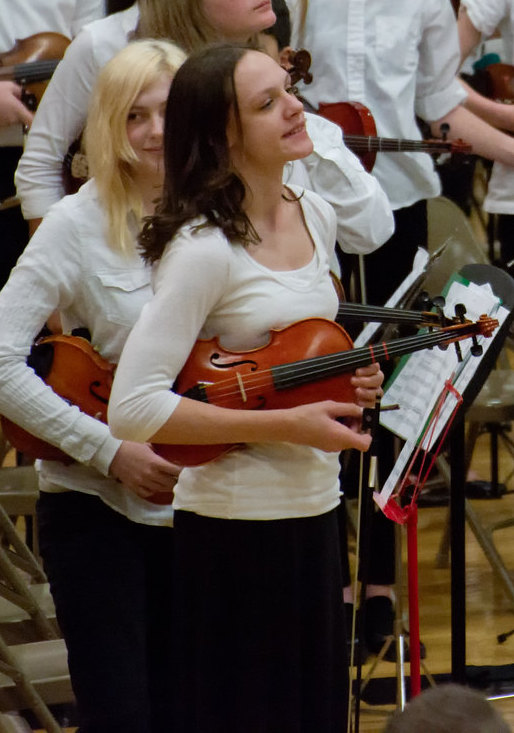}
\end{minipage}
\vspace{-5pt}
    \caption{Example input and output.}
    \label{fig:violin_prog}
\end{figure}
 
\mypar{Neuro-symbolic program synthesis} 
\revise{To generate the desired program from the user's demonstrations, our approach first represents the training images in  \emph{symbolic} form. In particular, rather than viewing each image as a set of low-level pixels, our approach generates a symbolic representation of each image, mapping object identifiers to their properties. This symbolic representation is obtained by running pre-trained neural networks for image segmentation and object recognition on the user-provided images.}

\revise{One of the key advantages of this symbolic image representation is that it allows reducing our complex learning task to the relatively well-understood \emph{programming-by-example (PBE)} problem. In particular, by representing the image in this symbolic form, \toolname can keep track of  which actions  have been  applied to which objects. Hence, the learning task reduces to synthesizing a so-called \emph{extractor} that can be used to programmatically extract the \emph{desired} objects among \emph{all} the objects comprising the symbolic image.}



\revise{While our proposed symbolic image representation allows reducing the learning problem to standard PBE, the resulting PBE task is unfortunately quite challenging. In particular, because images often contain a very large number of objects, the search space for the underlying synthesis problem can become quite massive. To make matters worse, each object has a large number of attributes associated with it, where each attribute corresponds to a pre-trained classifier (e.g., for detecting whether someone is smiling, whether an object is a guitar, etc.). Because each attribute corresponds to a built-in function in the underlying DSL, this means that the space of all programs that the synthesizer needs to consider can be enormous.
}


 \revise{As described briefly in Section~\ref{sec:intro}, the PBE technique underlying our synthesis engine is based on top-down enumerative search, but it utilizes two novel ideas to deal with the scalability challenges that arise in this setting:}

\mypar{Idea \#1: Combining term rewriting with partial evaluation} Despite the large search space, it turns out that many of the programs in our DSL are redundant. For example, consider the partial programs $P_1 = $ \textsf{Union}({\sf Is}(\textsf{Face}({\tt Id})), $\hole$)  and $P_2 = $ \textsf{Union}($\hole$, {\sf Is}(\textsf{Face}({\tt Id}))) where $\hole$ indicates a \emph{hole} (i.e., unknown subprogram). Since the \textsf{Union} operator is commutative,  any solution to the synthesis problem that is a completion of $P_2$ will also be a completion of $P_1$. Thus, we can significantly reduce the search space by detecting such redundant partial programs and pruning them from the search space. To that end, our method uses term rewriting to reduce each partial program to a canonical form and discards all  non-canonical expressions when performing the search. This idea can be seen as an instance of \emph{equivalence reduction} explored in prior work~\cite{aws_equiv_reduc}.   


While the above idea is quite useful in our setting, it is nonetheless \emph{not} sufficient to detect all redundant partial programs of interest. In particular, while two partial programs may not be \emph{always} equivalent, they might still be \emph{observationally equivalent} --- that is, they are guaranteed to have the same behavior on the \emph{given} set of input images. To gain more intuition, consider the  partial extractor $E = $ \textsf{Union}(\textsf{Intersect}({\sf Is}(\textsf{Smiling}), {\sf Is}(\textsf{EyesOpen})), $\hole$) whose canonical form is itself.
However, suppose that the example images provided by the user do not contain any human faces that are both smiling and have their eyes open. Under this assumption, \textsf{Intersect}({\sf Is}(\textsf{Smiling}), {\sf Is}(\textsf{EyesOpen})) is the empty set, so $E$ simplifies to $\hole$. 
Motivated by this observation, our method combines \emph{partial evaluation} with \emph{term rewriting} to further reduce the search space. In particular, our method partially evaluates incomplete programs on the provided input-output examples before reducing them to a canonical form. Because partial evaluation can greatly simplify incomplete programs, combining term writing with partial evaluation significantly  amplifies the pruning power of this technique. \revise{We believe this insight (namely,  combining partial evaluation with term rewriting) could be similarly powerful in reducing the search space in other synthesis settings beyond our image editing domain.}

\mypar{Idea \#2: Goal-directed reasoning via image abstractions} Our synthesis method uses another key idea, namely \emph{goal-directed reasoning via abstraction}, to successfully automate image editing tasks of interest. For example, consider the partial program \textsf{Intersect}($\hole_1$, $\hole_2$), and suppose the goal output $o$ is the set of all dog objects in the example images. While we cannot infer the \emph{exact} output of each hole, we \emph{can} infer that the $\hole_1$ and $\hole_2$ \emph{must} both produce all dog objects in the image due to the semantics of set intersection. Using this kind of goal-directed reasoning, we can prune all partial programs where either hole is instantiated with  {\sf Is}(\textsf{Object(Cat)}) (or any other extractor that does not produce all dogs). Similarly, consider the partial program \textsf{Union}($\hole_1$, $\hole_2$), and suppose  that the target output is again the set of all dogs in the image (and nothing else). In this case, we can infer that each hole should \emph{not} produce anything other than a dog because of the semantics of set union. Hence, if either hole is instantiated with {\sf Is}(\textsf{Object(Cat)}) (or any extractor that produces a non-dog object), we know that the program will be infeasible and can be safely pruned. 

\revise{Based on this motivation, our synthesis algorithm performs a form of abstract interpretation over images to facilitate goal-directed reasoning. In particular, starting from the desired output image, our synthesis technique utilizes the abstract semantics of the image editing DSL to infer specifications of sub-programs yet to be synthesized. These specifications take the form of pairs of over- and under-approximations, $(\absimg^-, \absimg^+)$, where $\absimg^+$ includes all objects that \emph{may} be present in the synthesized program and $\absimg^-$ represents those objects that \emph{must} be present. If a sub-program with inferred specification $(\absimg^-, \absimg^+) $ ever produces a symbolic image that contains fewer objects than its over-approximation or more objects than its under-approximation, we know it must be incorrect. In our setting, this idea can be used to prune large parts of the search space, and we believe that our proposed abstraction could  be similarly useful in other program synthesis tasks involving images.}

\section{Domain-specific language for image manipulation}

In this section, we introduce our domain-specific language for image manipulations. Since inputs to programs in this DSL are images, we first explain how we represent images and then describe the constructs in this DSL.

\begin{figure}
\begin{minipage}{0.60\textwidth}
    \includegraphics[width=0.7\textwidth]{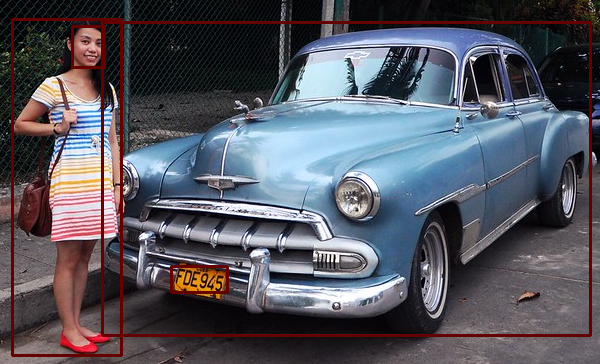}
\end{minipage}
\begin{minipage}{0.39\textwidth}
    \small
    $\absimg = \{(\prop_1, \bbox_1), (\prop_2, \bbox_2), (\prop_3, \bbox_3), (\prop_4, \bbox_4)$\} \\
    $\prop_1$ = \{ \texttt{objectType} $\rightarrow$ \texttt{person} \} \\
    $\prop_2$ = \{ \texttt{objectType} $\rightarrow$  \texttt{face}, \\ 
    \hspace*{20pt} \texttt{faceId} $\rightarrow$ \texttt{1}, \\ 
    \hspace*{20pt} \texttt{Smiling} $\rightarrow$ \texttt{true}, \\
    \hspace*{20pt} \texttt{EyesOpen} $\rightarrow$ \texttt{true}  \} \\
     $\prop_3$ = \{ \texttt{objectType} $\rightarrow$ \texttt{car} \} \\
    $\prop_4$ = \{ \texttt{objectType} $\rightarrow$ \texttt{text}, \\
    \hspace*{20pt} \texttt{textBody} $\rightarrow$ \texttt{"FDE945"} \} 
\end{minipage}
    \caption{An image and its corresponding symbolic image. Rectangles denote  bounding boxes.}
    \label{fig:abs_img_example}
\end{figure}

\mypar{Image representation}  A raw image $\image$ as a $n \times m$ matrix where each entry corresponds to a pixel. \revise{However, because raw images are quite low level, this work utilizes a more abstract representation called a \emph{symbolic image} for formalizing our DSL.} Specifically, given a raw image $\image$, we define a corresponding \emph{symbolic image} $\absimg$ as follows:

\begin{definition}{\bf (Symbolic image)}\label{def:image} A \emph{symbolic image} $\absimg$ is a set of \emph{objects} $o$ where each object is represented by a pair $(\prop, \bbox)$ where  $\prop$ is a mapping from the attributes of that object to their values, and $\bbox$ represents the location of the object within the raw image. For simplicity, we represent $
\bbox$ as a bounding box ($\idx_{\tt l}, \idx_{\tt r}, \idx_{t}, \idx_{b}$)  describing the left, right, top and bottom pixels. \end{definition}


Intuitively, a symbolic image $\absimg$ corresponds to a more abstract representation of $\image$ obtained through pre-trained neural networks used for classification (to construct $\prop$) and segmentation (used to construct $\bbox$). In particular, each element in the domain of $\prop$ is obtained using a different pre-trained neural network used for  classification. For example, consider an attribute called {\tt objectType} in the domain of $\prop$. This attribute identifies the type of the object, which could be a face, cat, dog, table etc. Some of the attributes are only defined for certain types of objects. For example, the boolean attribute {\tt Smiling} only makes sense for human faces. Thus, the domain of $\prop$ may be different across different objects. 
Given an object $o = (\prop, \bbox)$,  we use the notation $o.\prop$ and $o.\bbox$ to refer to $\prop$ and $\bbox$ respectively.

 \begin{example}
Consider the image in Figure \ref{fig:abs_img_example}, and its corresponding symbolic image $\absimg$ on the right. Here, $\absimg$ contains four objects: the person, their face, the car, and the text on the car's license plate. Each object has an attribute called \texttt{objectType}. The face object has the additional attributes \texttt{faceId}, \texttt{Smiling}, and \texttt{EyesOpen}, and the text object has the additional attribute \texttt{textBody} whose value is a string that contains the text on the license plate.  
 \end{example}

\revise{In the remainder of this paper, we use a single symbolic image to represent multiple raw input images. Because a symbolic image is a mapping from object identifiers to their attributes, a symbolic image can represent multiple raw images without any loss of information. This is because different occurrences of the same object in different images have different identifiers, and an attribute is used to track which object identifier originates from which image. This design choice of representing multiple raw images as a single symbolic image allows simplifying our technical presentation.}




\begin{figure}
    \centering
    \small
    \[
    \begin{array}{r l}
        \prog := & \{ \extractor \shortrightarrow \action, \cdots, \extractor \shortrightarrow \action\} \\
         \action := & {\sf Blur} \ | \ {\sf Blackout} \ | \ {\sf Sharpen} \ | \ {\sf Brighten} \ | \ {\sf Recolor} \ | \ {\sf Crop} \\
         \extractor := & {\sf All} \ | \ \revise{{\sf Is}(\attr)} \\
         | &  {\sf Complement}(\extractor) \ | \ {\sf Union}_N(E_1, \cdots, E_N) \ | \ {\sf Intersect}_N(E_1, \cdots, E_N) \\
         | & {\sf Find}(\extractor, \attr, \func) \ | \ {\sf Filter}(\extractor, \attr)  \\
          \attr := & {\sf Face}(N) \ | \ {\sf Object}(O) \ | \ {\sf Smiling} \ | \   \ {\sf AboveAge}(N) \ | \ {\sf Text}(W)  \ | \ \cdots \\
         \func := & {\sf GetLeft} \ | \  {\sf GetRight} \ | \  {\sf GetAbove} \ | \   {\sf GetBelow} \ |    \ {\sf GetParents} 
    \end{array}
    \]
    \vspace{-10pt}
    \caption{Image manipulation DSL. }
    \label{fig:dsl}
    \vspace{-0.2in}
\end{figure}

\mypar{DSL Syntax}
Our image editing DSL is defined in Figure \ref{fig:dsl} and is meant to  capture a broad class of selective edits
. At the top level, a program  is comprised of a set of guarded actions of the form $E \shortrightarrow A$, where $A$ is an action like {\sf Crop} or {\sf Blur} and  the guard $E$ is an \emph{extractor} that specifies what part of the image to apply that action to. Extractors are defined recursively and have two base cases: (1) the identity extractor \textsf{All} returns the entire image, and (2)  $\mathsf{Is}(\attr)$ returns all objects in the image  for which the predicate $\attr$ evaluates to true.  Extractors can be nested inside one another by composing them via set operators (complement, intersection, and union) as well as the constructs ${\sf Find}(E, \attr, f)$ and $\mathsf{Filter}(E, \attr)$.  The {\sf Find} construct first extracts a set of objects $O$ using the nested extractor $E$ and then, for each object $o \in O$, it  returns the first element in $f(o)$ satisfying predicate $\attr$. Here, $f$ is a function that takes as input an object and returns a sorted list of objects. For example, $\mathsf{GetRight}(o)$ returns a list of all objects that are to the right of $o$, sorted by their $x$ coordinate. Hence, the extractor $\mathsf{Find}(\mathsf{Is}(\mathsf{Face}(n)), \mathsf{Smiling}, \mathsf{GetRight})$ finds the first smiling face to the right of person $n$. As another example, $\mathsf{Find}(\mathsf{All}, \mathsf{Smiling}, \mathsf{GetLeft})$ would yield the set of all smiling faces that are to the left of \emph{some} object in the input image. Finally, the construct $\mathsf{Filter}(E, \attr)$ filters nested objects that satisfy predicate $\attr$. In particular, given a set of objects $O$ extracted via $E$,  $\mathsf{Filter}(E, \attr)$ returns all objects satisfying $\attr$  contained inside some object $o$ in $O$. For instance, the extractor $\textsf{Filter}(\textsf{Is}(\textsf{Object}(\texttt{car})), \textsf{Object}(\texttt{person}))$ will return all people who are inside of cars.

\begin{example}
Consider the image on the left in Figure \ref{fig:example_prog}. Given this image as input, the program 
\begin{align*}
\small
\begin{split}
\{\textsf{Inter}&\textsf{section}( \\ &\textsf{Find}(\textsf{Object}(\texttt{cat}), \textsf{Object}(\texttt{cat}), \textsf{GetRight}), \\
 &\textsf{Find}(\textsf{Object}(\texttt{cat}), \textsf{Object}(\texttt{cat}), \textsf{GetLeft})) \rightarrow \textsf{Blur}\}
 \end{split}
 \end{align*}
 will output the image on the right. Note that the extractor in this program yields all cat objects that have a cat  to their left and right. In other words, it extracts all cats that are between two other cats. 
\end{example}

\revise{Predicates in our DSL reflect the capabilities of state-of-the-art computer vision models for object recognition and classification. In particular, we choose to include certain predicates, such as $\mathsf{Face}(n)$ or $\mathsf{Object}(\texttt{cat})$, but not others (e.g., $\mathsf{Angry}$, $\mathsf{Sad}$), because existing neural networks are good at detecting the first class of features but not the latter class. Furthermore, the choice of built-in functions (e.g., $\textsf{GetLeft}$, $\textsf{GetAbove}$) is motivated by performing segmentation at the level bounding boxes. The remaining constructs  are either standard set operations (e.g., $\mathsf{Union}$) or well-understood functional combinators (e.g., $\mathsf{Filter}$). }


\begin{figure} 
\begin{minipage}[c]{0.35\textwidth}
\hspace*{30pt}
    \includegraphics[width=0.7\textwidth]{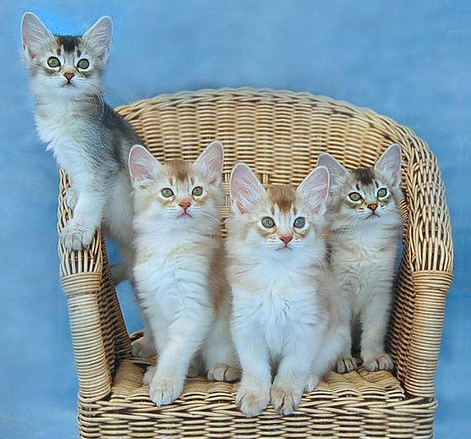}
\end{minipage}
\begin{minipage}[c]{0.02\textwidth}
\centering
\hspace*{8pt}
$\Longrightarrow$
\end{minipage}
\begin{minipage}[c]{0.35\textwidth}
\hspace*{30pt}
    \includegraphics[width=0.7\textwidth]{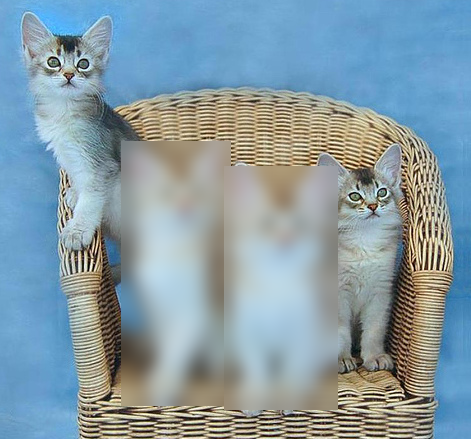}
\end{minipage}
\begin{minipage}[c]{0.25\textwidth}
\tiny
        \{\textsf{Intersection}( \\
                \hspace*{15pt} \textsf{Find}( \\
                \hspace*{25pt} \textsf{Object}(\texttt{cat}), \\
                \hspace*{25pt} \textsf{Object}(\texttt{cat}),  \\
                \hspace*{25pt} \textsf{GetRight}), \\
                \hspace*{15pt} \textsf{Find}( \\
                \hspace*{25pt} \textsf{Object}(\texttt{cat}), \\
                \hspace*{25pt} \textsf{Object}(\texttt{cat}), \\
                \hspace*{25pt} \textsf{GetLeft})) \\
             \hspace*{40pt} $\rightarrow$ \textsf{Blur}\} \\
\end{minipage}
    \caption{Input and output of a program.}
    \label{fig:example_prog}
\vspace*{-15pt}
\end{figure}

\begin{figure}
    \centering
    \small
    \begin{mathpar}
    \inferrule{ R \in \{ {\sf Smiling}, {\sf AreEyesOpen}, \cdots \} \\\\
    R \in {\sf Domain}(\imgobj.\prop) \ \ \ \ \imgobj.\prop[R] = {\sf True} }{\imgobj \models R} \ \  \ \ 
    \inferrule{R \in \{ {\sf BelowAge}, {\sf IsObject}, \cdots \} \\\\
    R \in {\sf Domain}(\imgobj.\prop) \ \ \ \ \imgobj.\prop[R] = C}{\imgobj \models R(C)}
    \end{mathpar}
    \vspace{-0.1in}
    \caption{Definition of the entailment relation.}
    \label{fig:attr}
    \vspace*{-0.5cm}
\end{figure}

\mypar{DSL Semantics}
The formal semantics of this DSL are presented in Figure~\ref{fig:semantics}. Given a program $P$ and input image $\image$, $\semantics{P}(\image)$ produces a new image $\image'$ by applying each of the actions in $P$ to the extracted sub-image. Similarly, given an  extractor $E$ and symbolic image $\absimg$, $\semantics{E}(\absimg)$ returns a set of objects contained in $\absimg$. Because each object stores its corresponding pixels in the original input image, note that it is trivial to convert a set of objects to pixels of the original image.  

\begin{figure}
    \centering
    \small
    \[
\begin{array}{r l}
    \bigsemantics{\extractor \shortrightarrow \action}(\image) = & {\sf ApplyAction}(\image, \action, \bigcup_{\imgobj \in \semantics{\extractor}(\absimg)} \imgobj.\bbox )\\
    \semantics{\extractor \shortrightarrow \action, \ \prog}(\image) = &  {\sf ApplyAction}(\bigsemantics{\prog}(\image), \action, \bigcup_{\imgobj \in \semantics{\extractor}(\absimg)} \imgobj.\bbox)\\
    \bigsemantics{{\sf All}}(\absimg) = & {\absimg}\\
    \bigsemantics{{\sf Is}(\attr)}(\absimg) = & \{ \imgobj \ | \ \imgobj \in \absimg \wedge \imgobj \models \attr \} \\
     \bigsemantics{{\sf Union}(\overline{E})}(\absimg) = & \bigcup_{\extractor \in \overline{E}} \bigsemantics{\extractor}(\absimg)\\
     \bigsemantics{{\sf Intersect}(\overline{E})}(\absimg) = & \bigcap_{\extractor \in \overline{E}} \bigsemantics{\extractor}(\absimg)\\
     \bigsemantics{{\sf Complement}(\extractor)}(\absimg) = & {\absimg} \setminus  \bigsemantics{\extractor}(\absimg) \\
     \bigsemantics{{\sf Filter}(\extractor, \attr)}(\absimg) =  & {\tt flatten}({\tt map}(\bigsemantics{\extractor}(\absimg), \lambda x. \ {\tt filter}(\bigsemantics{{\sf GetContents}}(x, \absimg), \attr)))) \\
     \bigsemantics{{\sf Find}(\extractor, \attr, \func)}(\absimg) =  & {\tt map}(\bigsemantics{\extractor}(\absimg), \lambda x. \ \bigsemantics{\func_\attr(x)}(\absimg)) \\ 
     \bigsemantics{f_\attr(\imgobj)}(\absimg) = & 
        \begin{cases}
            S[i]  & {\tt if \ } \exists_{0 \leq i < |S|} {S[i]} \models \attr  \bigwedge \forall_{0 \leq j < i} \ {S[j]} \not\models \attr \\
            {\sf None} & {\tt otherwise}
        \end{cases}
      \ \ \ {\tt where \ } S =  \bigsemantics{f}(\imgobj, {\absimg}) \\
\end{array}
    \]
    \vspace{-10pt}
    \caption{DSL semantics. Here,  ${\tt map}(S, f)$ takes the input $S: {\sf Set}[T]$, $f: T \shortrightarrow {\sf Option}[T]$ and returns $\{f(s) \ | \ s \in S \wedge f(s) \neq {\sf None}\}$. {\tt flatten} takes in a set of sets $S_{\sf All} = \{S_1, \cdots S_n\}$ and returns a set that containing $\bigcup_{S_i \in S_{\sf All}} \cup_{s \in S_i} s$. Finally, {\tt filter} is the standard filter operator.}
    \label{fig:semantics}
    \vspace{-10pt}
\end{figure}

The  semantics of  extractors  are defined in terms of symbolic images introduced in Definition~\ref{def:image}. In particular, given an image $\image$, $\mathsf{Is}(\attr)$ returns the set of all objects in $\absimg$ that satisfy $\attr$. As defined in Figure~\ref{fig:attr}, an object $o$ satisfies a predicate of the form $R(C)$, denoted $o \models R(C)$,  if $o.\prop$ contains an attribute called $R$ and the  value of that attribute is $C$. Similarly, if $R$ is a  nullary relation, we have $o \models R$ iff $o$ has an attribute called $R$ whose value is \emph{true}. 

\begin{figure}
    \centering
    \small
    \[
    \begin{array}{r l}
    \bigsemantics{{\sf GetRight}}(\imgobj, \absimg) = & {\sf Sort}(\{ \imgobj' \ | \ \imgobj' \in \absimg \wedge \imgobj'.\bbox[\idx_{\sf l}] \ge \imgobj.\bbox[\idx_{\sf l}] \}, o.\bbox[\idx_{\sf {l}}]) \\
     \bigsemantics{{\sf GetLeft}}(\imgobj, \absimg) = & {\sf SortReverse}(\{ \imgobj' \ | \ \imgobj' \in \absimg \wedge \imgobj'.\bbox[\idx_{\sf r}] \le \imgobj.\bbox[\idx_{\sf r}] \}, o.\bbox[\idx_{\sf {r}}]) \\
     \bigsemantics{{\sf GetAbove}}(\imgobj,\absimg) = & {\sf Sort}(\{ \imgobj' \ | \ \imgobj' \in \absimg \wedge \imgobj'.\bbox[\idx_{\sf t}] \ge \imgobj.\bbox[\idx_{\sf t}] \}, o.\bbox[\idx_{\sf {t}}]) \\
     \bigsemantics{{\sf GetBelow}}(\imgobj, \absimg) = & {\sf SortReverse}(\{ \imgobj' \ | \ \imgobj' \in \absimg \wedge \imgobj'.\bbox[\idx_{\sf b}] \ge \imgobj.\bbox[\idx_{\sf b}] \}, o.\bbox[\idx_{\sf {b}}]) \\
     \bigsemantics{{\sf GetParents}}(\imgobj, \absimg) = & {\sf Sort}(\{ \imgobj' \ | \ \imgobj' \in \absimg \wedge \bigsemantics{{\sf Contains}}(\imgobj'.\bbox, \imgobj.\bbox) \}, {\sf GetSize}(\bbox)) \\
     \bigsemantics{{\sf GetContents}}(\imgobj, \absimg) = & [ \imgobj' \ | \ \imgobj' \in \absimg \wedge \bigsemantics{{\sf Contains}}(\imgobj.\bbox, \imgobj'.\bbox) \wedge  \\
     & \not\exists_{o'' \in \absimg}. \ o'' \neq o \wedge \bigsemantics{{\sf Contains}}(o.\bbox, o''.\bbox) \wedge \bigsemantics{{\sf Contains}}(o''.\bbox, o'.\bbox)  ] \\
     \bigsemantics{{\sf Contains}}(\imgobj.\bbox, \imgobj'.\bbox) = & {\sf True} \text{ if } \\
     & \imgobj'.\bbox[j_{\sf l}] \geq \imgobj.\bbox[j_{\sf l}] \wedge \imgobj'.\bbox[j_{\sf r}] \leq \imgobj.\bbox[j_{\sf r}] \\
     &  \wedge \imgobj'.\bbox[j_{\sf t}] \geq \imgobj.\bbox[j_{\sf t}] \wedge
     \imgobj'.\bbox[j_{\sf b}] \leq \imgobj.\bbox[j_{\sf b}] \text{ else } {\sf False} \\
     \bigsemantics{{\sf GetSize}}(\bbox) = & (\bbox[j_{\sf r}] - \bbox[j_{\sf l}]) * (\bbox[j_{\sf b}] - \bbox[j_{\sf t}])
    \end{array}
     \]
     \vspace{-10pt}
    \caption{Semantics for the built-in and auxiliary functions $f$ in the DSL. ${\sf Sort}(S, {\sf key})$ sorts the objects in set $S$ from smallest to largest with respect to {\sf key}. ${\sf SortReverse}$ does the opposite. }
    \label{fig:func}
    \vspace*{-0.5cm}
\end{figure}

Since the semantics of set operators are standard, we only explain the semantics of  $\mathsf{Filter}$ and $\mathsf{Find}$, which are defined in terms of  functional combinators like {\tt map} and {\tt flatten}. Recall that the $\mathsf{Find}$ extractor is parameterized over a function $f$, such as $\mathsf{GetRight}$ and $\mathsf{GetBelow}$,  whose semantics are given in Figure~\ref{fig:func}. In particular, $\semantics{\mathsf{GetX}}(o, \absimg)$ yields a list of all objects $o'$ in $\absimg$ satisfying the \emph{spatial} relationship $X(o', o)$. As expected, the semantics of these functions are defined using the bounding box $\bbox$ of each object in the image. For example, given an object $o$ in $\absimg$, $\mathsf{GetRight}$ decides which objects are to the right of $o$ based on the leftmost pixels of the bounding box of each object. 

As shown in Figure~\ref{fig:semantics}, the $\mathsf{Find}$ construct first evaluates its nested extractor $E$ on the input image $\absimg$ to obtain a set of objects $O$ and applies the function $f_\varphi$ to each object  $o \in O$. The semantics of $f_\varphi(x)$ are given at the very bottom of Figure~\ref{fig:semantics} and essentially yield the first object satisfying $\varphi$ in  the list given by $f(x)$. Since $f(x)$ may be the empty list or may not have any elements satisfying $\varphi$, observe that $f_\varphi(x)$  can yield $\mathsf{None}$, which is discarded when constructing the output of $\mathsf{Find}$. 

Finally, we explain the semantics of the $\mathsf{Filter}$ construct, which first evaluates its nested extractor $E$ on the input image $\absimg$ to obtain a set of elements of $O$. Then, for each object $o \in O$, it obtains elements nested  inside of $o$ (by inspecting the bounding box of each object in the image) and only retains those elements that satisfy $\attr$. The final output of $\mathsf{Filter}$  is obtained by flattening the resulting set of sets into a single set.

\section{Problem Statement}\label{sec:statement}

In this section, we define the synthesis problem that we address in the remainder of the paper.  We first start by introducing the concept of an \emph{image edit}:

\begin{definition}{\bf (Edit)} Given an image $\image$, an edit $\edit$ on that image is a mapping from objects in $\absimg$ to a list of actions that have been applied to those objects.
\end{definition}
Given an image $\image$ and edit $\edit$, we use the notation $\image[\edit]$ to denote the resulting image obtained by applying $\edit$ to $\image$. 
The specification for our synthesis problem is defined in terms of edits:

\begin{definition}{\bf (Spec)} 
An \emph{image manipulation specification} $\imdspec$ is a mapping from images to edits. 
\end{definition}

We  now formally state our synthesis problem as follows:

\begin{definition}{\bf (Image manipulation by demonstration (IMBD))}
Given an image manipulation specification $\imdspec$, the goal of \emph{image manipulation by demonstration} is to produce a program $P$ in the DSL from Figure~\ref{fig:dsl} such that $
\forall (\image, \edit) \in \imdspec. \  P(\image) = \edit[\image] 
$.

\end{definition}

\section{Synthesis Algorithm}\label{sec:synthesis}

\begin{figure}[!t]
\vspace*{-0.5cm}
    \centering
    \small
    \begin{algorithm}[H]
    \begin{algorithmic}[1]
    \Procedure{Synthesize}{$\imdspec$}
    \Statex\Input{specification $\imdspec$}
    \Statex\Output{a program $\prog$ such that $\forall (\image, \edit) \in \imdspec. \  \prog(\image) = \edit[\image]$}
    \State $\prog \assign \emptyset$ 
    \ForAll{$\action \in \mathsf{Actions}$}
    \State $\absimg_{in} \assign  \bigcup_{(\image, \edit) \in \imdspec} \absimg$
    \State $\absimg_{out} \assign \{\imgobj \ | \ (\image, \edit) \in \imdspec \wedge \imgobj \in {\sf Domain}(\edit) \wedge \action \in \edit[\imgobj] \}$
    \If{$\absimg_{out} \neq \emptyset$}
    \State $\extractor \assign$ {\sc SynthesizeExtractor}$(\absimg_{in}, \absimg_{out})$
    \If{$\extractor = \bot$}
    \State \Return $\bot$
    \Else 
    \State $\prog \assign \prog \cup \{ \extractor \shortrightarrow \action \} $
    \EndIf
    \EndIf
    \EndFor
    \State \Return $\prog$
    \EndProcedure
    \end{algorithmic}
    \end{algorithm}
    \vspace{-0.5in}
    \caption{Top-level synthesis algorithm}
    \label{fig:top-level}
    \vspace*{-0.5cm}
\end{figure}


In this section, we describe our synthesis algorithm for solving the IMBD problem defined in the previous section. Our top-level learning procedure is shown in Figure~\ref{fig:top-level} and works as follows. For each possible action in the DSL, it constructs an input-output example $(\absimg_{in}, \absimg_{out})$ based on the specification $\imdspec$. If, for some action $\action$,  $\absimg_{out}$ is the empty set, this means that $\action$ is irrelevant to the target task, so the algorithm moves on to the next action. Otherwise, it invokes the {\sc SynthesizeExtractor} procedure on $(\absimg_{in}, \absimg_{out})$  to learn the corresponding extractor $E$ for   $\action$ and adds the guarded action $E \shortrightarrow \action$ to the synthesized program.

\subsection{Preliminaries}
As is evident from this discussion, the central part of our technique is the extractor learning algorithm, which relies on a particular representation of \emph{partial programs}:

\begin{definition}{\bf (Partial program)}
A partial program $\prog$ is a tree $(V, E, \mstmt, \mgoal)$ with nodes $V$ (including a special root node $v_0$) and directed edges $E$. The mapping $\mstmt$ maps each node in $V$ to a label, which is either a construct in our image manipulation DSL (e.g., \textsf{All}, \textsf{Complement}) or the special symbol $\hole$ representing a hole. Each node $v \in V$ is also annotated with a \emph{goal}  $\nspec$ such that $\mgoal(v) = \nspec$. We write $\prog \vdash v: (\nlabel, \nspec)$ to denote that $\mstmt(v) = \nlabel$ and $\mgoal(v) = \nspec$.   If none of the nodes in $\prog$ is labeled with a hole, we refer to $\prog$ as a \emph{complete program}. 
\end{definition}

\begin{example} 
Figure \ref{fig:ast} depicts the partial program \textsf{Union(Is(Smiling), $\hole$)} as a tree, where each node is annotated with its corresponding label. 

\end{example}

The goal annotation $\mgoal(v)$ for  each node $v$ in a partial program imposes constraints on the semantics of the subtree rooted at $v$.  In our context, a goal is defined as follows:

\begin{definition}{\bf (Goal annotation)}
A \emph{goal annotation} (or \emph{goal} for short) of a node in the partial program is a pair $(\absimg^-, \absimg^+)$ where $\absimg^-, \absimg^+$ are symbolic images \revise{corresponding to over- and under-approximations of the output}.
\end{definition}

Next, we define the consistency between a symbolic image and a goal as follows:

\begin{definition}{\bf (Consistency with goal)}\label{def:consistent}
We say that a symbolic image $\absimg$ is consistent with a goal $\nspec = (\absimg^-, \absimg^+)$, denoted $\absimg \sim \nspec$, iff 
$\absimg^- \subseteq \absimg \subseteq \absimg^+$.
\end{definition}

We also extend this notion of consistency to partial programs:

\begin{definition}{\bf (Consistency of partial program)}\label{def:consistent}
A partial program $\prog$ is consistent with a symbolic image $\absimg$ iff,  for every complete subtree $\prog_v$ of $P$ rooted at node $v$, we have $\semantics{\prog_v}(\absimg) \sim \Pi(v)$.  
\end{definition}

Intuitively, the goals annotating a partial program are used for guiding extractor synthesis and for ensuring that we never enumerate inconsistent partial programs. 

\begin{example}
Consider again the partial program from Figure~\ref{fig:ast}, which contains a complete subprogram, namely \textsf{Is}(\textsf{Smiling}), rooted at node $v_1$. Suppose that we have an input-output example $(\absimg_{in},\absimg_{out})$, where $\absimg_{in}$ contains several face objects. Then $\semantics{\textsf{Is}(\textsf{Smiling})}(\absimg_{in})$ will be a symbolic image $\absimg'$ containing just the faces that are smiling. Further, suppose our output symbolic image $\absimg_{out}$ contains just the faces that are smiling or have their eyes open. The goal annotation of $\textsf{Is}(\textsf{Smiling})$ will be $\Pi(v_1) = (\emptyset, \absimg_{out})$, as explained later in Section \ref{sec:goalinference}. Since $\emptyset \subseteq \absimg' \subseteq \absimg_{out}$, we have $\semantics{\textsf{Is}(\textsf{Smiling})}(\absimg_{in}) \sim \Pi(v_1)$. Therefore, this partial program is consistent. To illustrate inconsistent partial programs, now consider  \textsf{Union}(\textsf{Is}(\textsf{Object}(\texttt{cat})), $\hole$) and the same input-output example $(\absimg_{in}, \absimg_{out})$. Here,  $\semantics{\textsf{Is}(\textsf{Object}(\texttt{cat}))}(\absimg_{in})$ will be a symbolic image $\absimg'$ containing all cat objects. The goal annotation of $\textsf{Is}(\textsf{Object}(\texttt{cat}))$ will again be $\Pi(v_1) = (\emptyset, \absimg_{out})$. Since $\absimg' \not\subseteq \absimg_{out}$, $\semantics{\textsf{Is}(\textsf{Object}(\texttt{cat}))}(\absimg_{in}) \not\sim \Pi(v_1)$. Therefore, this partial program is inconsistent.

\end{example}

Because our synthesis algorithm gradually replaces holes with concrete programs, we conclude this section by defining an operation to update partial programs:

\begin{definition}{\bf (Partial program update)}
Given a partial program $\prog = (V, E, \mstmt, \mgoal)$, we use the notation $\prog[v_0 \annot (\nlabel_0, \nspec_0), \ldots v_n \annot (\nlabel_n, \nspec_n) ]$ to indicate the new partial program \[
\prog' = (V \cup \bigcup_{i=0}^ n \{ v_i\}, \  E \cup \bigcup_{i=1}^n \{(v_0, v_i)\}, \mstmt[v_0 \mapsto \nlabel_0, \ldots v_n \mapsto \nlabel_n], \mgoal[v_0 \mapsto \nlabel_0, \ldots v_n \mapsto \nlabel_n])\]

\end{definition}

In other words, the notation $\prog[v_0 \annot (\nlabel_1, \nspec_1), \ldots v_n \annot (\nlabel_n, \nspec_n) ]$  corresponds to adding children $(v_1, \ldots, v_n)$ of $v_0$ (and their corresponding labels and goals) and updating the label and goal of $v_0$.  Finally, we write $\mathsf{CreateProg}(v, \nlabel, \nspec)$ to denote the creation of a partial program with a single node $v$ with label $\nlabel$ and goal annotation $\nspec$.



\subsection{Top-Level Extractor Learning Algorithm}\label{sec:top-level}

\begin{figure}[!t]
    \centering
    \small
    \vspace*{-0.5cm}
    \begin{algorithm}[H]
    \begin{algorithmic}[1]
    \Procedure{SynthesizeExtractor}{$\absimg_{in}, \absimg_{out}$}
    \Statex\Input{ $\absimg_{in}$ is an input symbolic image and $\absimg_{out}$ is the output symbolic image}
    \Statex\Output{an extractor program $\prog$ such that $\bigsemantics{\prog}(\absimg_{in}) \equiv \absimg_{out}$}
    \State $\worklist \assign \{ \mathsf{CreateProg}(v_0, \hole, (\absimg_{out}, \absimg_{out})) \}$
    \While{$\worklist \neq \emptyset$}
    \State $\prog \assign \worklist.{\sf remove}()$
    \If{{\sf isComplete}($\prog$)}
    \If {$\bigsemantics{\prog}(\absimg_{in}) \equiv \absimg_{out}$} \Return $\prog$
    \EndIf
    \Else 
    \State $v \assign {\sf SelectOpenNode}(\prog)$;
    \ForAll{$\prog' \in \textsc{Expand}(\prog, v)$}
    \State $\prog''\assign ${\sc PartialEval}$(\prog', \absimg_{in})$
    \If{$\prog'' \neq \bot  \ \wedge \  \neg ${\sc Reducible}$(\prog'')$}
    \State $\worklist \assign \worklist \cup \{\prog'\}$;
    \EndIf
    \EndFor
    \EndIf
    \EndWhile
    \State \Return $\bot $
    \EndProcedure
    \end{algorithmic}
    \end{algorithm}
    \vspace{-0.5in}
    \caption{Extractor synthesis algorithm. }
    \label{fig:extractor}
    \vspace*{-0.5cm}
\end{figure}

We now present our top-level extractor learning algorithm, which is shown in Figure~\ref{fig:extractor}. \revise{Given an input symbolic image $\absimg_{in}$ and an output symbolic image $\absimg_{out}$,} this algorithm maintains a worklist $\worklist$ of partial programs and iteratively adds to this list. At the beginning of the procedure, $\worklist$ is initialized to  a single program with one node $v_0$ and no edges. Node $v_0$ has label $\hole$ and goal output $\phi_0 = (\absimg_{out}, \absimg_{out})$. The loop in lines 3-12 dequeues a program $\prog$ from the worklist and processes it. The worklist keeps programs in ascending order first by AST size, then by AST depth. If $\prog$ is complete and satisfies the correctness condition, the procedure terminates and returns $P$. Otherwise, \textsc{SynthesizeExtractor} calls \textsc{Expand} on line 9 to generate a new set of partial programs by expanding an open node $v$ in $P$. As we describe in Section \ref{sec:goalinference}, the expansion procedure also infers goals for each new hole in the partial program. 

\begin{wrapfigure}{r}{0.4\textwidth}
\vspace{-25pt}
        \includegraphics[width=0.4\textwidth]{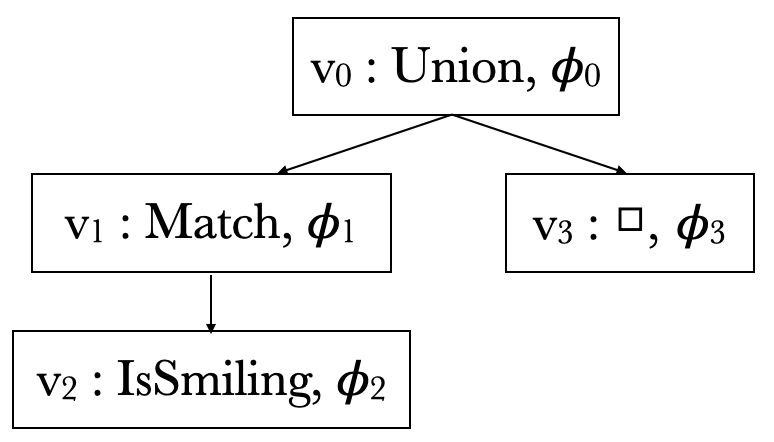}
\vspace{-10pt}
    \caption{An illustration of a partial program.}
    \label{fig:ast}
    \vspace{-10pt}
\end{wrapfigure}


Next, for each expansion $P'$ of $P$, the algorithm calls \textsc{PartialEval} on line 10 to generate a partially evaluated program $P''$. The partial evaluation procedure identifies each complete subprogram $P_v$ of $P'$, evaluates it on the input image $\absimg_\emph{in}$ to obtain an output image $\absimg$, and replaces $P_v$ with the constant $\absimg$. As we discuss in more detail in Section \ref{section:partialeval}, {\sc PartialEval} can return $\bot$ if it finds any inconsistent subprograms; in this case, \revise{$\prog'$} is \emph{not} added to the worklist.


If partial evaluation does not return $\bot$, the  algorithm calls the \textsc{Reducible} procedure on line 11, which is used to check whether $P''$ can be simplified.  Since {\sc SynthesizeExtractor} explores programs in increasing order of complexity, we know that $P''$ is redundant if  \textsc{Reducible} returns $\top$. Hence, the algorithm  adds $P''$ to the worklist only if $P''$ cannot be further simplified. 



\subsection{Goal Inference} \label{sec:goalinference}

As mentioned earlier, a key component of our extractor learning approach is the inference of goals for each node. This goal inference method is presented in Figure~\ref{fig:expand} as part of the {\sc Expand} procedure. Every time the algorithm expands a hole  associated with node $v$, it picks an (n-ary) DSL operator $\sf{f}$, updates $v$'s label to $\sf{f}$, and adds $n$ new children $v_1, \ldots, v_n$ of $v$. 
Each child node $v_i$ of $v$ is marked as being ``open" (i.e., labeled with a hole) and is annotated with its corresponding goal. Observe that goal inference is performed using the function $\infgoal{\sf{f}}$, which takes as input a goal annotation $\nspec$ and produces a new goal $\nspec_f$ for the arguments of the function $\sf{f}$.

\begin{figure}[!t]
    \centering
    \small
    \begin{mathpar}
\inferrule{
  P \vdash v: (\hole, \nspec) \ \ \ \ \ \nspec_f = \infgoal{{\sf f}}(\nspec)
   }
  {{\textsc{Expand}}(P, v) = \{\prog[v \annot ({\sf f}, \nspec), v_1 \annot (\hole, \nspec_f),  \ldots, v_n \annot (\hole, \nspec_f)] \ | \ {\sf f} \in \mathcal{F}, \overline{v} \text{ fresh}\} } 
  \\\\
  \end{mathpar}
  \vspace{-30pt}
    \[
  \begin{array}{rl}
        \infgoal{{\sf Union}}(\absimg^-, \absimg^+) & = (\emptyset, \absimg^+)\\
        \infgoal{{\sf Intersect}}(\absimg^- \absimg^+) & = (\absimg^-, \absimg_{in})\\
        \infgoal{{\sf Complement}}(\absimg^-, \absimg^+) & = (\absimg_{in}\setminus\absimg^+, \absimg_{in} \setminus \absimg^-) \\
        \infgoal{{\sf FilterContents}}(\absimg^-, \absimg^+) & = (\emptyset, \absimg_{in}) \\
        \infgoal{{\sf Find}}(\absimg^-, \absimg^+) & = (\emptyset, \absimg_{in}) 
  \end{array}
  \]
  \vspace{-10pt}
    \caption{Inference rules for {\sc Expand}. $\absimg_{in}$ is the input symbolic image, and $\mathcal{F}$ represents DSL functions.}
    \label{fig:expand}
    \vspace*{-0.5cm}
\end{figure}

Recall that a goal annotation is of the form $(\absimg^-, \absimg^+)$  where $\absimg^-$ and $\absimg^+$ are symbolic images under- and over-approximating the image objects associated with the subprogram rooted at that node. In more detail, if a node $v$ has the goal annotation $(\absimg^-, \absimg^+)$, the consistency requirement from Definition~\ref{def:consistent} stipulates that, in order for $P$ to be consistent with the input symbolic image $\absimg_{in}$, the subprogram $P_v$ rooted at node $v$ must produce a set of objects that is a superset of $\absimg^-$ when executed on the input image $\absimg_{in}$. Similarly, it also requires that $\semantics{P_v}(\absimg_{in})$ is a subset of $\absimg^+$. Hence, given a goal $\nspec$ on the output of a DSL operator $\mathsf{f}$, goal inference aims to propagate under- and over-approximations to each of $\mathsf{f}$'s arguments. Put simply, the goal annotations approximate the output that a program must have in order for its parent program to also have a valid output, so programs that do not match their goal annotation can be safely pruned from the search space.  We formalize this notion in the following theorem.\footnote{Proofs are provided in the Appendix of the extended version of the paper ~\cite{imageeye}.}

\begin{theorem} \label{thm:consistency}
Let $P$ be a partial program derived by \textsc{SynthesizeExtractor} whose root node has goal annotation $(\absimg_{out}, \absimg_{out})$. If $P$ is not consistent with $\absimg_{in}$, then for any completion $P'$ of $P$, $\bigsemantics{P'}(\absimg_{in}) \not\equiv \absimg_{out}$.
\end{theorem}

The proof of this theorem  crucially relies on the correctness of the goal inference rules,  which we explain in more detail next. \\

\noindent
{\bf \emph{Union.}} Consider the DSL expression $E = \mathsf{Union}(E_1, \ldots, E_n)$, and suppose that the goal annotation for this expression is $(\absimg^-, \absimg^+)$. Since the over-approximation for the whole expression is $\absimg^+$, the operands of $\mathsf{Union}$ should not produce objects that are not in $\absimg^+$. Hence, the over-approximation for each operand is also $\absimg^+$. In other words, if \emph{any} operand outputs an object $o$ that is not in $\absimg^+$, then $E$ will output $o$ as well, which is not valid. In contrast, the only safe under-approximation we can infer for the operands is $\emptyset$, as there is no particular object $o$ that each operand \emph{must} output in order for $E$ to output all objects in $\absimg^-$.
\\

\noindent
{\bf \emph{Intersect.}} Consider the expression $E = \mathsf{Intersect}(E_1, \ldots, E_n)$, and suppose that the goal annotation for this expression is $(\absimg^-, \absimg^+)$. By the semantics of \textsf{Intersect}, for $E$ to output each $o \in \absimg^-$, each operand $E_i$ must also output $o$. Thus, the under-approximation for each $E_i$ is also $\absimg^-$. In contrast, we cannot deduce anything about elements that must \emph{not} be in the output of any $E_i$; thus, the over-approximation for the operands is the entire input image $\absimg_{in}$. \\

\noindent
{\bf \emph{Complement.}} Consider the expression $E = \mathsf{Complement}(E')$. If $E$ must produce an object $o$ (i.e., $o \in \absimg^-)$, then $E'$ must not produce it. Hence, the over-approximation for $E'$ is $\absimg_{in}\backslash \absimg^-$. In contrast, if $E$ must not produce an object $o$ (i.e., $o \not \in \absimg^+$), then $o$ must be produced by $E'$. Hence, the under-approximation is $\absimg_{in} \backslash \absimg^+$. \\

\noindent 
{\bf \emph{Find, Filter.}} In the case of the {\sf Find} and {\sf Filter} constructs, we cannot propagate meaningful approximations to the nested extractors, resulting in the trivial goal annotation $(\emptyset, \absimg_{in})$. To gain intuition about why this is the case, consider the expression $\mathsf{Find}(E', \varphi, {\sf GetLeft})$. For any object $o \in \absimg_{in}$, it could be the case that there is no object located to the left of $o$ in image $\image_{in}$, meaning that $\bigsemantics{\textsf{GetLeft}}(o, \absimg_{in})$ will be empty. In other words, if $o$ is output by $E'$, it will have no impact on the output of $E$. Hence, any object in $\absimg_{in}$ \emph{could} be output by $E'$, which is why the over-approximation is $\absimg_{in}$. For similar reasons, we also cannot infer any sound under-approximation other than $\emptyset$.



\begin{example}\label{ex:goal_infer}
Consider the image $\absimg$ from Figure \ref{fig:abs_img_example}, and  let $\absimg_{out} = \{(\prop_4, \bbox_4)\}$ be the output symbolic image containing only  the license plate. Now consider the partial program:
\[
\textsf{Union}(\textsf{Complement}(\textsf{Is}(\textsf{Object}(\texttt{car}))), \hole)
\]
whose top-level goal is $(\absimg_{out}, \absimg_{out})$. Since the subprogram $\textsf{Complement}(\textsf{Is} (\textsf{Object}(\texttt{car})))$ is an operand of a \textsf{Union}, it has goal $(\emptyset, \absimg_{out})$. Further, the subprogram   $\textsf{Is}(\textsf{Object}(\texttt{car}))$ is the operand of a \textsf{Complement}, so it has goal $(\{ (\prop_1, \bbox_1), (\prop_2, \bbox_2), (\prop_3, \bbox_3) \}, \absimg)$.

\end{example}



\subsection{Partial Evaluation} \label{section:partialeval}

As stated earlier, our synthesis algorithm performs partial evaluation to amplify the power of goal-directed reasoning as well as  equivalence reduction. In particular, given a partial program $\prog$, the {\sc PartialEval} procedure invoked  in Figure~\ref{fig:extractor} returns another partial program $\prog'$ by evaluating the complete subprograms of $\prog$ on the input. Partial evaluation can also reveal that  $\prog$ is infeasible; in  this case, {\sc PartialEval} returns $\bot$ to indicate that $\prog$ violates  consistency (Definition~\ref{def:consistent}).

\begin{figure}
    \centering
    \small
    \begin{mathpar}
    \inferrule*[Left=Hole]{{\sf Root}(\prog) = v \\ \prog \vdash v : (\hole, \nspec) 
    }
    {\absimg_{in} \vdash \prog \leadsto \prog}\\
    \inferrule*[Left=Const]{{\sf Root}(\prog) = v \\ \prog \vdash v : (\absimg, \nspec) 
    }
    {\absimg_{in} \vdash \prog \leadsto \prog}\\
\inferrule*[Left=Complete]{
    {\sf IsComplete}(\prog) \\ \bigsemantics{\prog}(\absimg_{in}) = \absimg \\
    {\sf Root}(\prog) = v \\ \prog \vdash v : (l, \nspec)}{\absimg_{in} \vdash \prog \leadsto (\absimg \sim \nspec) \ ? \ {\sf CreateProg}(v, \absimg, \nspec)  : \bot} \\
\inferrule*[Left=Partial]{\neg{\sf IsComplete}(\prog)  \\ {\sf Root}(\prog) = v \\ \prog \vdash v : (l, \nspec)\\\\
   {\sf Children}(\prog, v) = \{v_1, \ldots, v_n \}  \\ \absimg_{in} \vdash {\sf Subtree}(\prog, v_i) \leadsto \prog_i 
    }{\absimg_{in} \vdash \prog \leadsto (\forall_i. \ \prog_i \neq \bot) \ ? \ \prog[\prog_1/{\sf Subtree}(\prog, v_1), \cdots, \prog_n / {\sf Subtree}(\prog, v_n)]: \bot}\\
    \end{mathpar}
    \vspace{-35pt}
    \caption{Rules for {\sc PartialEvaluation}. $\prog[\prog_i/{\sf Subtree}(\prog, v_i)]$ represents replacing the subprogram of $\prog$ rooted at node $v_i$ with the new subprogram $\prog_i$.}
    \label{fig:partialeval}
    \vspace{-10pt}
\end{figure}

We present our {\sc PartialEval} procedure using the inference rules summarized in  Figure \ref{fig:partialeval}. The first rule, labeled {\sc Hole}, states that open nodes cannot be evaluated, as they represent a completely unconstrained program. The second rule, labeled {\sc Const}, states that constants simply evaluate to themselves. The third rule, labeled {\sc Complete}, evaluates complete subprograms by executing them on the input. If the resulting output $\absimg$ is inconsistent with the goal annotation $\nspec$, partial evaluation yields $\bot$; otherwise, it produces the constant $\absimg$. The final rule, labeled {\sc Partial},  applies to incomplete programs and recursively applies {\sc PartialEval} to each subprogram rooted at the root node. If any of these subprograms are inconsistent, then the whole program is also inconsistent, and the algorithm returns $\bot$. Otherwise, it constructs a new partial program where each subprogram $P_i$ of the root node is replaced with its  partially evaluated version $P_i'$. 

\begin{example}
Consider the program $\textsf{Union}(\textsf{Complement}(\textsf{Is}( \textsf{Object}(\texttt{car}))), \hole)$
from Example \ref{ex:goal_infer} and the desired output image containing just the license plate (i.e., $\{ (\prop_4, \bbox_4)\}$). This program is incomplete, so the \textsc{Partial} rule will recursively apply \textsc{PartialEval}. The subprogram $\textsf{Complement}(\textsf{Is}(
\textsf{Object}(\texttt{car}))$ is complete, so the \textsc{Complete} rule will evaluate this subprogram on the input symbolic image $\absimg$ to obtain $\absimg' = \{ (\prop_1, \bbox_1), (\prop_2, \bbox_2), (\prop_4, \bbox_4) \}$. Recall also (from Example \ref{ex:goal_infer}) that the goal of this subprogram is $ (\emptyset, \{(\prop_4, \bbox_4)\})$. Since $\absimg' \not\subseteq \{ (\prop_4, \bbox_4) \} $, $\absimg'$ is not consistent with the goal, so partial evaluation will return $\bot$.
Intuitively, this program should be pruned because, no matter how we instantiate the hole, the top-level program will always produce objects (e.g., the human face) that are not part of the desired output image. 

\end{example}

\subsection{Equivalence Reduction} \label{section:equivreduction}

We conclude this section by describing our equivalence reduction technique for identifying redundant partial programs. In particular, recall that a partial program $\prog'$ is \emph{redundant} with respect to another partial program $\prog$ if, for every completion $C'$ of $\prog'$, there is a corresponding completion $C$ of $\prog$ such that $C$ and $C'$ produce the same output on the input examples.  In other words, because such partial programs $\prog, \prog'$ are observationally equivalent on the inputs of interest, it suffices to merge them into one equivalence class. Thus, our technique can be viewed as extending the notion of \emph{
observational equivalence} from complete to partial programs. \begin{figure}
\centering
\[
    \small
    \begin{array}{r l}
     {\sf Union}(A, A) \rewrite A & {\sf Intersect}(A, A)  \rewrite A \\
    {\sf Union}(A, {\sf Intersect}(A, B)) \rewrite A & {\sf Intersect}(A, {\sf Union}(A, B)) \rewrite A \\
        \textsf{Union}(A, B) \rewrite B \text{ if $A \subseteq B$.} & 
     \textsf{Intersect}(A, B) \rewrite A \text{ if $A \subseteq B$.} \\
    {\sf Complement}({\sf Complement} (A)) \rewrite A & {\sf Union}(B, A)  \rewrite {\sf Union}(A, B) \\ 
    {\sf Intersect}(B, A) & \rewrite {\sf Intersect}(A, B) \\
    {\sf Union}({\sf Complement}(A), {\sf Complement}(B)) & \rewrite {\sf Complement}({\sf Intersect}(A, B)) \\ 
    {\sf Intersect}({\sf Complement}(A), {\sf Complement}(B)) & \rewrite {\sf Complement}({\sf Union}(A, B)) \\ 
    {\sf Union}({\sf Intersect}(A, B), {\sf Intersect}(A, C)) & \rewrite {\sf Intersect}(A, {\sf Union}(B, C)) \\ 
    {\sf Intersect}({\sf Union}(A, B), {\sf Union}(A, C)) & \rewrite {\sf Union}(A, {\sf Intersect}(B, C))  
    \end{array}
\]
\vspace{-15pt}
    \caption{Rewrite rules. }
    \label{fig:rewrite-rules}
    \vspace{-10pt}
\end{figure}

At a high level, there are two key components of our equivalence reduction technique: (1) partial evaluation (already discussed in Section~\ref{section:partialeval}) and (2) term rewriting. Given a partially evaluated program $\prog$, our synthesis algorithm checks whether it is possible to simplify $\prog$ using a set of rewrite rules that capture known equivalences between expressions in our DSL. Figure~\ref{fig:rewrite-rules}
 shows the rewrite rules for our DSL using the notation $l \rewrite r$, meaning that a term that matches $l$ can be rewritten into the form on the right. Observe that the free variables in $l$ and $r$ are universally quantified, so a term $t$ is said to match the left-hand-side $l$ if there exists a substitution $\sigma$ such that $t = l[\sigma]$. Furthermore, the result of applying this rewrite rule to $t$ is $r[\sigma]$.

With this notation in place, we now turn our attention to the {\sc Reducible} procedure called by the {\sc SynthesizeExtractor} algorithm. Recall that {\sc Reducible} returns a boolean ($\top$ or $\bot$) to indicate whether a term can be simplified using a set $\rws$ of domain-specific rewrite rules. This {\sc Reducible} procedure is defined using the two inference rules shown in Figure~\ref{fig:reducible}. According to the first rule ({\sc Base}), holes and constant values are not reducible. The second rule labeled {\sc Rec} deals with terms $E \equiv {\sf{f}}(E_1, \ldots, E_n)$ by recursively invoking the {\sc Reducible} procedure on each $E_i$. If any $E_i$ is reducible, it also returns reducible. Otherwise, it checks whether any rewrite rule $\rw \in \rws$ matches  $E$, meaning that the left-hand side of $\rw$ can be unified with $E$. If so, it returns true, and false otherwise.

 \begin{figure}
    \centering
    \small
    \begin{mathpar}
    \inferrule*[Left=Base]{{\sf Root}(\prog) = v \\ \prog \vdash v: (l, \nspec) \\ l \in \{\hole, \absimg\}}{\rws \vdash \prog \reduce \bot} \\ 
    \inferrule*[Left=Rec]{{\sf Root}(\prog) = v \\ \prog \vdash v: ({\sf f}, \nspec)\\\\
   \in {\sf{Children}}(\prog, v) = \{ v_1, \ldots, v_n\} \\ \rws \vdash {\sf Subtree}(\prog, v_i) \reduce b_i
    }{\rws \vdash \prog \reduce (\exists_i. \ b_i = \top \vee \exists \rw \in \rws. {\sf Is}(\prog, \rw)) \ ? \ \top : \bot}\\
    \end{mathpar}
    \vspace{-0.5in}
    \caption{Inference rules for {\sc Reducible}.   $\rws$ represents all rewrite rules, some of which are shown in Figure~\ref{fig:rewrite-rules}. }
    \label{fig:reducible}
    \vspace*{-0.5cm}
\end{figure}

\begin{example}

Consider a partial program of the form $\textsf{Union}(P_1, P_2, \hole)$ where $P_1, P_2$ have been partially evaluated as $\absimg_1$ and $\absimg_2$, respectively.  Suppose that the symbolic image $\absimg_1 $ is the set of objects $\{\imgobj_1, \imgobj_2, \imgobj_3\}$ and $\absimg_2$ is $\{\imgobj_2, \imgobj_3\}$. Since $\absimg_2 \subseteq \absimg_1$, this program will match with the rewrite rule 
\begin{align*}
    \textsf{Union}(A_1, \ldots, A_i, \ldots, A_n) \rewrite \textsf{Union}(A_1, \ldots, A_n) \text{ if $\exists j$ such that $A_i \subseteq A_j$.}
\end{align*}
which corresponds to the domination rule for sets.  Thus, this program simplifies to $\textsf{Union}(P_1, \hole)$, meaning that the \textsc{Reducible} procedure will return $\top$ and this partial program will be pruned. 
\end{example}

\section{Implementation}


We have implemented the proposed algorithm as a new tool called \toolname written in Python. In what follows, we describe key implementation details that are not covered in the technical sections.

\mypar{Computer vision primitives} Recall that our DSL operates over symbolic images, which are generated from the raw input image by applying existing computer vision primitives. In our implementation, we use the Amazon Rekognition library 
for object classification, text detection, and facial attribute classification. Compared with similar vision libraries, Rekognition offers more capabilities that are well-suited for image manipulation tasks of interest to this work.


\mypar{Graphical user interface} \toolname also incorporates a graphical user interface that allows users to demonstrate the desired image processing task. Our GUI is implemented in  JavaScript and supports both  image manipulation as well as image search. 
To use the GUI, the user first uploads their batch of images
 and then selects one or more images to annotate.  For each image being annotated, the GUI indicates  regions of the image that are classified as an object with rectangular bounding boxes. In the image editing mode, the user can select one of these objects and then apply the desired action (e.g., crop, blur, or highlight). When using \toolname in  search mode, the user can  indicate the image as either being of interest or irrelevant. 
 Once the user is done annotating a representative set of images, they press a button to invoke the synthesizer. If synthesis is successful,  \toolname  applies the generated  program to the entire image set and uploads the output to a new directory,
 which contains all the relevant images with the desired edits applied to them. 

\section{Evaluation}

In this section, we describe the results of our experimental evaluation, which is designed to answer the following research questions:
\begin{itemize}[leftmargin=*]
    \item \textbf{RQ1.} Can \toolname automate interesting image manipulation and exploration tasks?
    \item \textbf{RQ2.} How many examples does \toolname need to synthesize the intended program?
    \item \textbf{RQ3.} How does \toolname's synthesis algorithm compare against  existing baselines?
    \item \textbf{RQ4.} How important are the pruning techniques used by the synthesizer?
    \revise{
    \item \textbf{RQ5.} How effective are the synthesized programs in producing the desired edit on the test set?}

\end{itemize}

\mypar{Benchmarks} To answer these questions, we collected a set of 50 benchmark tasks across three domains, namely Wedding, Receipts, and Objects.  Tasks in the Wedding domain involve identifying and manipulating specific faces. An example task in this domain is to “crop out wedding guests who are not smiling.” Tasks in the Receipts domain involve identifying specific words or classes of text, such as “highlight the prices to the right of the words ‘total’ and ‘subtotal.’” Tasks in the Objects domain require manipulating specific classes of objects that are spatially related to other objects, such as, "crop the faces of people playing the guitar". Many of these tasks are motivated by real-world scenarios found on image editing forums, such as Reddit groups related to  Photoshop and GIMP. For each task, we manually wrote a ground truth program in our DSL that can be used to check the correctness of the program returned by \toolname.  

\begin{figure}
    \begin{table}[H]
\footnotesize
\caption{Statistics about  images and tasks for each domain. Program size is measured in terms of AST nodes.}
\vspace{-10pt}
\begin{tabular}{|c|c|c|c|c|}
    \hline
     {\bf Dataset}  & {\bf \#  Images} & {\bf Avg. \#  Objects} & {\bf \# Tasks} & {\bf Avg. Program Size} \\
     \hline
     Wedding & 121 & 10 & 16 & 9.4 \\
     \hline
     Receipts & 38 & 59 & 13 & 7.8  \\
     \hline
     Objects & 608 & 3 & 21 & 8.3 \\
     \hline
\end{tabular}
\label{tab:dataset}
\vspace*{-20pt}
\end{table}
\end{figure}

Table \ref{tab:dataset} gives some statistics about each of the three domains used in our evaluation. As we can see, each domain varies in terms of the number of images they contain and the average number of objects in a given image. Observe that the Receipts domain contains the largest number of objects per image because each word is identified as a unique text object.  In contrast, images in the Objects domain are much more sparse. For each domain, we have between 13 and 21 synthesis tasks, and the average size (in terms of AST nodes) of the ground truth program is in the 8-10 range.

\subsection{Experimental Setup}\label{setup}

To answer our first research question, we attempted to use \toolname to automate each of our 50 benchmark tasks using the following methodology: We first select an image from the task's domain and apply the desired edit. When choosing an image, we prefer those that contain as few objects as possible, as this choice involves the least amount of work for  the user.  Then, we use \toolname to synthesize a program based on  this single demonstration. If the generated program produces the desired edit on \emph{all} images in the data set, we consider the task to be successfully automated. Otherwise, we select a single image where  \toolname does \emph{not} produce the desired edit and re-attempt synthesis with this \emph{additional} example. We continue this process for up to 10 rounds and up to 180 seconds per round.  All of our experiments are conducted on a desktop machine with 2.3 GHz dual-core Intel core i5 CPU and 8 GB of physical memory.

\subsection{Main Results}\label{sec:main}

Table~\ref{tab:results} presents the results of this experiment. The key takeaway is that \toolname can successfully 
automate 48 of the 50 tasks in our benchmark suite within the given resource limits.  Table~\ref{tab:results}  also shows average and median synthesis times for the last round of user interaction. As we can see from this table, average  synthesis time is around 15 seconds, with the median being much faster at around 1 second. We also note that synthesis time varies significantly across the domains, with the fastest being Objects and slowest being Receipts. This discrepancy makes sense considering the average  number of objects per domain. In particular, recall that the number of constants in the DSL depends on the number of objects in the target domain, so synthesis generally takes longer in domains like Receipts that contain a lot of objects. However, the Receipts domain generally requires fewer rounds of user interaction, as object-dense images are richer in information. The last column of Table~\ref{tab:results} shows the average number of rounds of user interaction. As we can see, the average number of demonstrations required across all three domains is just below 4. 

\begin{figure}
\vspace{2mm}
    \begin{table}[H]
\footnotesize
\caption{Summary of results for \toolname. \revise{We include 95\% confidence intervals.}}
\vspace{-10pt}
\begin{tabular}{|c|c|c|c|c|}
    \hline
     {\bf Dataset}  & {\bf \# solved} & {\bf Avg. Synth Time (s)} & {\bf Med. Synth Time (s)} & {\bf Avg. \# Examples} \\
     \hline
     Wedding & 14/16 & 15.6 \revise{$\pm$ 13.4} & 5.5 & 5.4 \revise{$\pm$ 1.0} \\
     \hline
     Receipts & 13/13 & 25.4 \revise{$\pm$ 23.4} & 1.6 & 2.2 \revise{$\pm$ 0.65} \\
     \hline
     Objects & 21/21 & 3.2 \revise{$\pm$ 2.4} & 0.1 & 3.8 \revise{$\pm$ 0.5} \\
     \hline
     {\bf Total} &  {\bf 48/50} & {\bf 12.8 \revise{$\pm$ 8.0}} & {\bf 1.2} & {\bf 3.8 \revise{$\pm$ 0.5}} \\ 
     \hline 
\end{tabular}
\vspace{-20pt}
\label{tab:results}
\end{table}
\end{figure}

\mypar{Failure analysis} We now examine the two tasks that \toolname fails to successfully automate. One of these tasks is from the Wedding domain and requires cropping the image to feature just the bride and the people standing directly to her left and right. In this case, \toolname fails to find the correct program within the time limit of 180 seconds because the size of the ground truth program is fairly large and there are a large number of detected objects.
The second task that \toolname fails to automate is also in the Wedding domain and involves identifying images that contain the bride's face only when there are people standing directly to her left and right. For this benchmark, \toolname requires more than 10 rounds of user interaction to find the desired program. 
Since this task requires extracting the bride's face only in a specific circumstance, there are many simpler programs that produce the same output on nearly all photos in the dataset. 

\vspace{0.1in}
\noindent\fbox{%
    \parbox{.98\textwidth}{%
        \textbf{Result for RQ1:} \toolname automates 48 out of 50 interesting image manipulation and exploration tasks, with a median synthesis time of 1.1 seconds.
    }%
}

\vspace{0.1in}
\noindent\fbox{%
    \parbox{.98\textwidth}{%
        \textbf{Result for RQ2:} \toolname requires an average of 4  images to synthesize the intended program.
    }%
}

\subsection{Comparison with Other Synthesis Tools}

To answer our third research question, we compare the synthesis engine of \toolname with existing synthesis tools. \revise{However, since existing  tools do not support the image editing domain, we first reduce our learning  problem to PBE (as discussed in Sections~\ref{sec:statement} and \ref{sec:top-level}). Furthermore, since prior work does not consider DSLs that operate over images, we   cast our synthesis problem as an instance of  \emph{syntax-guided synthesis} (SyGuS) and instantiate the SyGuS framework with our domain-specific language. Among the solvers that support the SyGuS format, we compare \toolname's synthesis engine against the  two most recent winners of the SyGuS competition. } One of these solvers \cite{cvc5} extends the CVC SMT solver~\cite{cvc4} to support syntax-guided synthesis. The second one, EUSolver~\cite{eusolver}, is based on bottom-up enumerative search with equivalence reduction and uses a divide-and-conquer approach to decompose the synthesis task into smaller problems. 

Among these solvers, we found the CVC solver to be ineffective at solving the synthesis problems that arise in our setting. In particular, instantiating our DSL in the CVC framework requires using the theory of sets (to represent symbolic images), but the resulting synthesis problems in this background theory are not easily solvable using a purely theorem proving approach. In fact, we found that this SMT-based approach is unable to solve even the simplest of our synthesis tasks within the given time limit.

In contrast, we were able to successfully instantiate EUSolver to handle the synthesis tasks from our image editing domain. The results of the comparison against EUSolver are presented in Figure~\ref{fig:comparison} as a bar graph. Here, the $x$-axis indicates the difficulty level of the synthesis tasks (as measured by AST size); thus,  bars in this plot correspond to synthesis tasks of increasing difficulty. On the other hand, the $y$-axis shows the number of tasks completed within the given time limit. The solid blue bars correspond to the results for \toolname, and the hatched orange bars correspond to those of EUSolver. As we can see from this figure, EUSolver can solve 14 out of 16 of the easiest tasks, but, as the difficulty level increases, there is a growing gap between \toolname and EUSolver. Overall, \toolname can solve 14 more tasks than EUSolver out of the 50 tasks total. 

\revise{To gain some intuition about these results, we briefly discuss why ImageEye outperforms EUSolver on our benchmarks. First, unlike EUSolver which is a generic solver, \toolname performs a form of abstract interpretation customized to images and our image editing DSL. This type of reasoning allows \toolname to prune many infeasible programs that need to be enumerated by EUSolver.   Second, many of the techniques in EUSolver target branching, but our DSL allows branching in a stylized manner (at the top level and as part of filtering constructs). Finally, EUSolver works by combining sub-programs that work on a subset of examples, and this particular decomposition strategy does not seem effective in the image domain. For these reasons, \toolname is more effective at solving the PBE problems that arise in the context of image extractor synthesis.}

\vspace{0.1in}
\noindent\fbox{%
    \parbox{.98\textwidth}{%
        \textbf{Result for RQ3:} The baseline synthesis tool, EUSolver, can successfully solve  68\% of the benchmarks compared with 96\% solved by \toolname.
    }%
}

\begin{figure}
\vspace{2mm}
\begin{minipage}{.45\textwidth}
    \centering
    \centering
    \includegraphics[width=0.95\textwidth]{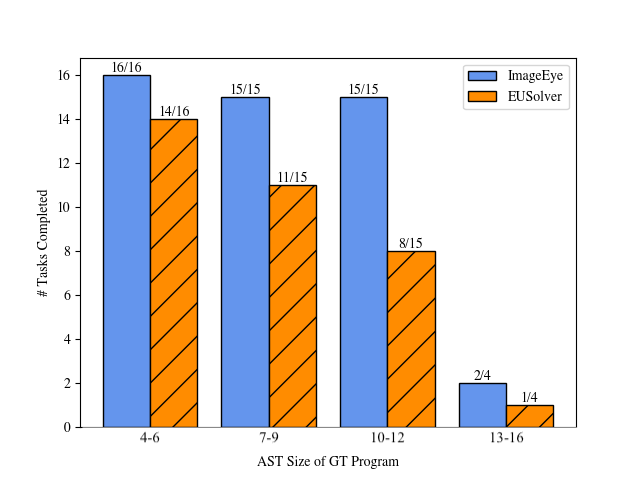}
    \vspace*{-0.4cm}
    \caption{Comparison of \toolname and EUSolver.}
    \label{fig:comparison}
\end{minipage}%
\begin{minipage}{.5\textwidth}
    \centering
    \centering
    \includegraphics[width=1.0\textwidth]{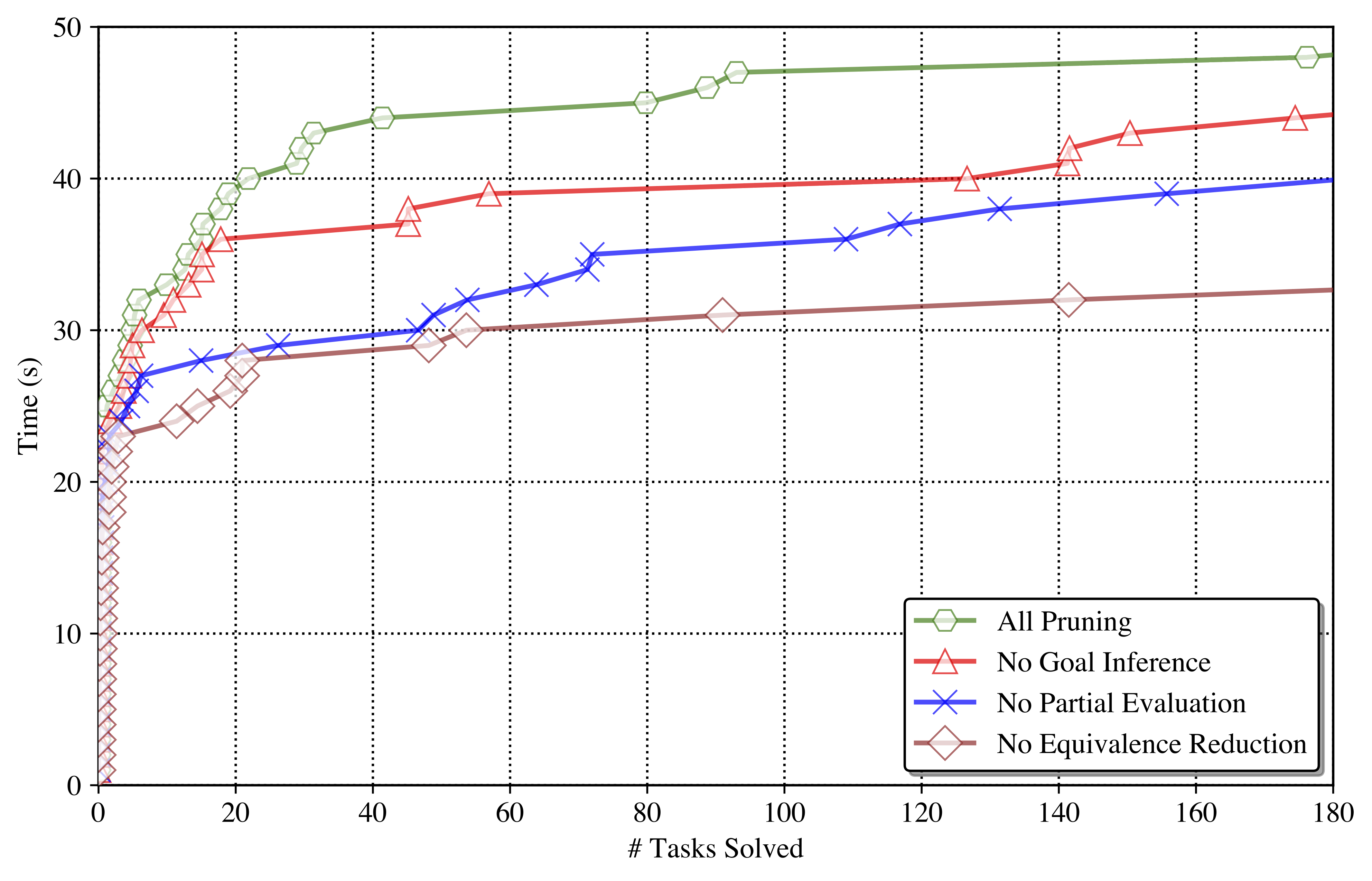}
    \vspace*{-0.7cm}
    \caption{Ablation study for \toolname}
    \label{fig:ablation}
\end{minipage}
\vspace*{-0.5cm}
\end{figure}

\subsection{Ablation Study}

To answer our final  research question, we present the results of an ablation study in which we disable some of the key components of our synthesis algorithm. In particular, we consider the following three ablations of \toolname:

\begin{itemize}[leftmargin=*]
    \item {\bf No Goal Inference:} This ablation does not use the goal inference technique of Section~\ref{sec:goalinference}. However, it does perform equivalence reduction with partial evaluation and term rewriting. 
    \item {\bf No Partial Evaluation:} This version of \toolname does not perform partial evaluation before applying the term rewrite rules from Section~\ref{section:equivreduction}. However, it does perform goal inference and uses rewrite rules to prune the search space. 
    
    \item {\bf No Equivalence Reduction:} This ablation does not perform equivalence reduction using term rewriting. In other words, it does not utilize the techniques described in Section~\ref{section:equivreduction}.
\end{itemize}

The results of this ablation study are presented as a cactus plot in Figure \ref{fig:ablation}. Here, the $x$-axis shows cumulative synthesis time and the $y$-axis shows the number of benchmarks solved within a given time. As we can see from this figure, all of our proposed techniques have a significant impact on synthesis time. Without goal inference, \toolname times out on four additional tasks and takes around 14 seconds longer on average to solve the tasks on which it does not time out. 
Without partial evaluation, \toolname times out on eight additional tasks and takes around 23 seconds longer on average. 
 Finally, without equivalence reduction, \toolname times out on 16 additional tasks.
 


\vspace{0.1in}
\noindent\fbox{%
    \parbox{.98\textwidth}{%
        \textbf{Result for RQ4:} The  techniques discussed in Sections~\ref{sec:goalinference}-\ref{section:equivreduction} are important for making synthesis effective in the image editing domain.
    }%
}

\subsection{Reliability of Underlying Neural Models}

\revise{When reporting our main experimental results in Section~\ref{sec:main}, we manually inspect the synthesized program and consider the synthesis result to be correct if it is semantically equivalent to the ground truth program we wrote by hand. However, because the synthesized programs contain neural networks for object recognition and classification, even a correct synthesized program may not produce the expected output for all images in the test set.  For instance, if the desired edit is to blur all cats in an image, and the object classification model does not recognize a specific cat in an image, then the program $\textsf{Is}(\textsf{Object}(\texttt{cat}))$ will not produce the desired output.}

\revise{In this section, we additionally evaluate the accuracy of the synthesized programs  in terms of  the percentage of images in the test set for which the desired output is produced. However, since there are a very large number of images in some of the data sets, we randomly sample 20 images from each of the three data sets\footnote{To ensure relevance of the sampled images, we re-sample if the output of the synthesized program is empty on that image.} and manually examine if the synthesized program produces the intended output for each of these 20 images. }


\revise{Overall, across the three domains, we find that the synthesized programs produce the intended output on \textbf{87\%} of the sampled images. Many of the failure cases stem from the same misclassification occurring numerous times. For instance, for the wedding data set, the face recognition model fails to identify that a specific wedding guest is smiling across many images.}

\vspace{0.1in}
\noindent\fbox{%
    \parbox{.98\textwidth}{%
      \revise{  \textbf{Result for RQ5:} The programs synthesized by \toolname produce the desired edit for 87\% of the images in the test set.}
    }%
}

\section{Limitations}

\revise{In this section, we discuss some of the main limitations of \toolname. First, the effectiveness of \toolname is highly dependent on the underlying neural components. For example, if the target task involves a class of objects that the model cannot reliably identify,  \toolname will not be effective in producing the intended output image. However, we try to mitigate this problem through the choice of the neural primitives included in the DSL and intentionally exclude object classifiers that do not work reliably in practice. Additionally, we note that manually refining a small portion of the images in the data set is preferable over editing all images manually.  }

\revise{A second limitation of \toolname is due to its user interaction model. In particular, to decide whether the synthesized program is correct, the user needs to inspect all images in the dataset, and, even then, it may be hard to distinguish whether any problems in the output are due to the lack of sufficient demonstrations or due to limitations of the  neural primitives in the synthesized program. Additionally, if the underlying neural primitives misclassify relevant objects in the \emph{training} example, then the user will not be able to perform their demonstration, as the \toolname GUI only allows editing objects recognized by the object recognition engine. This design choice  is intentional in that the interface \emph{forces} the user to perform demonstration on ``good" images. However, a potential disadvantage is that the user may need to go through multiple images before they find one  on which the demonstration can be performed. 
One way to address both of these limitations could be through  active learning approaches that suggest images for the user to label. }

\revise{Finally, \toolname is limited by the expressivity of the DSL, which  only provides built-in functions like $\textsf{GetAbove}$ and $\textsf{GetRight}$ that can be evaluated easily by using bounding boxes. These functions are not suitable for reasoning about three-dimensional spatial concepts like one object being behind another. This limitation could be addressed by extending the DSL with functions that are implemented using additional neural primitives.}

\section{Related work}

\mypar{Image manipulation} Image manipulation is a long-standing problem in computer vision, graphics, and computational photography. Recent efforts in this space have used deep neural networks to generate realistic variants of a given image, applying them to tasks  like inpainting~\cite{gatedcov, foreground, edgeconnect}, extrapolation~\cite{extrapolation1, extrapolation2}, and photo editing~\cite{stargan, modulargan, fader, brock2017neural}. As a representative example of a work in this space,   Fader~\cite{fader} can generate variants of a subject with different attributes  like age or gender. In this work, we solve a different type of image manipulation problem than most of these prior efforts: our focus is on identifying what operations to apply to \emph{which parts} of the image, rather than generating realistic variants of a given input image. Furthermore, our approach is based on neurosymbolic program synthesis rather than generative neural networks. However, these approaches can be incorporated into our overall approach by treating them as pre-trained neural network primitives in our DSL. 


\mypar{Neurosymbolic programming for images} Recently, there has been growing interest in using neurosymbolic DSLs that include both logical and neural components in the image domain~\cite{tian2018learning, pmlr-v97-young19a, mao2018the, pmlr-v119-huang20h, ellis18, Johnson_2017_ICCV, ReedF15}. Similar to our work, these efforts typically combine symbolic operators for higher-level reasoning with neural modules for perception. The closest work in this space is that of Huang et al. ~\cite{pmlr-v119-huang20h}, which generates programmatic \emph{referring expressions} that  identify specific objects  in terms of their attributes and relationships with respect to other objects in the image. However, their work differs from ours in several respects: (1) They focus on locating a single object whereas we focus on applying  actions to a \emph{set} of objects; (2) they synthesize logic programs using a different synthesis algorithm based on deep Q-learning and hierarchical search; and (3) their focus is on a synthetic dataset with geometric shapes whereas our focus is on more realistic images with faces, text, and arbitrary objects.

\mypar{Top-down enumerative search} Several recent synthesis techniques use a combination of top-down enumerative search and lightweight deductive reasoning to significantly reduce the search space~\cite{escher, neo, lambda2, flashmeta, myth, scythe, smyth}. Among these, our approach bears similarities to enumerative synthesis approaches that propagate the goal to the missing subexpressions. In particular, {\sc Myth}~\cite{myth}, {\sc SMyth}~\cite{smyth} and $\lambda^2$~\cite{lambda2} infer new input-output examples for the holes in a partial programs by utilizing type information embedded in the language. While our method also performs goal-directed reasoning, the underlying deductive reasoning techniques are different. 
Another synthesis framework that uses example-based specifications is {\sc FlashMeta}~\cite{flashmeta}, which propagates specifications from the DSL operators down into their arguments using so-called \emph{witness functions}. Unlike the synthesis algorithm we present here, {\sc FlashMeta} uses version space algebras (VSA) to represent the space of all programs that are consistent with the provided input-output examples.



Prior efforts on regular expression synthesis~\cite{regel, alpharegex, opsynth}  also utilize over- and under-approximations  to eliminate infeasible programs. In particular, {\sc Regel}~\cite{regel} and {\sc AlphaRegex}~\cite{alpharegex} both   derive  over- and under-approximations of the set of strings that could be matched by a partial regex. In contrast, we use over- and under-approximations in a  different context and  approximate the synthesis sub-goals as opposed to the outputs of a given partial program. 


\mypar{Synthesis using term rewriting} There has been several efforts that use term rewriting~\cite{rewrite} to speed up program synthesis~\cite{Dershowitz1993467, Reddy1989,aws_equiv_reduc, mitra}. These techniques have found applications in many  domains, including CAD model construction~\cite{szalinski}, robotic process automation~\cite{webrobot}, compiler construction~\cite{compiler_rewrite}, and writing numerical software~\cite{Boyle1997}.  Similar to the work of Smith et al.~\cite{aws_equiv_reduc}, we also use  an equational rewrite system to reduce the number of partial programs enumerated during top-down synthesis; however, our technique combines this idea with partial evaluation~\cite{partial_eval}  and goal-directed reasoning to make it more effective.

\mypar{Synthesis using partial evaluation} There are a variety of domain-specific~\cite{sketch, rosette} and domain-agnostic~\cite{morpheus, idips}  synthesis techniques that use partial evaluation~\cite{partial_eval} to obtain a more efficient synthesis procedure. In particular, {\sc Morpheus}~\cite{morpheus} utilizes partial evaluation to infer a  more precise specification of the partial program, which helps to  increase its SMT-based pruning power. Similar to our approach, both {\sc Rosette}~\cite{rosette} and IDIPS~\cite{idips} evaluate the concrete part of the partial program to obtain a simplified version. However, our work differs from these prior techniques in that we use partial evaluation to make term rewriting more effective. 


\mypar{Programming by demonstration} Programming-by-demonstration  techniques~\cite{pbd} utilize user demonstrations to learn a new task. This paradigm has been successfully adopted in a variety of  scenarios, including web automation~\cite{helena, webrobot, vegemite, ringer}, robot learning~\cite{robot_learn1, robot_learn2, robot_learn3}, text editing~\cite{smartedit}, and SQL query synthesis~\cite{sickle}. \toolname also allows users to demonstrate the desired task through a graphical user interface and leverages the demonstration to decompose the synthesis task into a set of PBE problems, one for each action in the demonstration.


\section{Conclusion}
We have presented a new synthesis-based approach for automating image editing and  search tasks. Given a few user demonstrations performed through a graphical user interface, our method synthesizes a program that can be used to automate the desired image search or batch editing task. At the heart of our approach lies a neuro-symbolic DSL that combines functional operators with pre-trained neural modules for object detection and classification.
 We have implemented this approach in a new tool called \toolname and evaluated it on 50 image search and editing tasks across three different domains involving human faces, text, and arbitrary objects. Our evaluation shows that \toolname can automate 96\% of these tasks, with a median synthesis time of 1 second and requiring on average four user demonstrations. 

\begin{acks}                            
  We would like to thank Michelle Ding, fellow graduate students on GDC 5S, and the anonymous reviewers for their help and feedback on this paper. This material is based upon work supported by the \grantsponsor{GS100000001}{National Science Foundation}{http://dx.doi.org/10.13039/100000001} under grant numbers \grantnum{GS100000001}{CCF-1811865} and \grantnum{GS100000001}{CCF-1918889}, Google under the Google Faculty Research Grant, as well as Facebook, Amazon, Intel, and RelationalAI.
\end{acks}

\section*{Data Availability}
An artifact supporting the results of this paper is available on Zenodo \cite{imageeye-artifact}. It includes the ImageEye implementation, benchmarks, and benchmarking scripts.

\bibliography{main}

\pagebreak
\appendix
\section{Proofs}

\setcounter{theorem}{0}

\begin{lemma} \label{lemma:consistency}
Let $P$ be a partial program derived by \textsc{SynthesizeExtractor}. If $P$ is not consistent with $\absimg_{in}$, then for any completion $P'$ of $P$, $\bigsemantics{P'}(\absimg_{in}) \not\sim \phi$.
\end{lemma}

\begin{proof}
By structural induction on $P$.

\textbf{Base Case 1:} $P = \hole$. There are no complete subtrees of $P$, so it is vacuously true that if $P$ is not consistent with $\absimg_{in}$, then for any completion $P'$ of $P$, $\bigsemantics{P'}(\absimg_{in}) \not\sim \phi$. 

\textbf{Base Case 2:} $P = \textsf{All}$ or $P = \textsf{Is}(\varphi)$. Suppose $P$'s root node has goal annotation $\phi = (\absimg^-, \absimg^+)$. Suppose also that $P$ is not consistent with $\absimg_{in}$. In either case, the only complete subtree of $P$ is $P$, so it must be that $\bigsemantics{P}(\absimg_{in}) \not\sim \phi$. Also, since $P$ has no holes, the only completion of $P$ is $P$. Then it must be the case that for any completion $P'$ of $P$, $\bigsemantics{P'}(\absimg_{in}) \not\sim \phi$.

\textbf{Inductive Hypothesis:} Assume this lemma holds for partial programs $P_1, \ldots, P_n$. We will show that for any program $P$ constructed from existing partial programs, the lemma holds for $P$.

\textbf{Inductive Case:} We will show that for any program $P$ constructed from existing partial programs, the lemma holds for $P$. Let $P$ have the goal annotation $\phi = (\absimg^-, \absimg^+)$, and suppose that $P$ is not consistent with $\absimg_{in}$. Then there exists a complete sub-program $P_v$ of $P$ such that $\bigsemantics{P_v}(\absimg_{in}) \not\sim \Pi(v)$. We will proceed by considering all the possibilities for the top-level construct of $P$. 

\begin{itemize}
    \item $P = \textsf{Find}(P_1, \varphi)$. Since $P$ was derived by \textsc{SynthesizeExtractor}, the goal annotation of $P_1$ is $\phi_1 = (\emptyset, \absimg_{in})$. Note that either $P_v = P$, or $P_v$ is a complete sub-program of $P_1$. Suppose the latter case. Then $P_1$ is not consistent with $\absimg_{in}$. Let $\textsf{Find}(P_1', \varphi)$ be a completion $P$, where $P_1'$ is a completion of $P_1$. By inductive hypothesis, $\bigsemantics{P_1'}(\absimg_{in}) \not\sim \phi_1$. Then either $\emptyset \not\subseteq \absimg_{in}$ or $\absimg_{in} \not\subseteq \absimg_{in}$. However, neither of these cases are possible. Thus, it must be that $P_v = P$. Then $P$ must be a complete program. The only completion of $P$ is $P$, so it must be the case that for any completion $P'$ of $P$, $\bigsemantics{P'}(\absimg_{in}) \not\sim \phi$.

\item $P = \textsf{Filter}(P_1, \varphi, f)$. This proof is identical to that of \textsf{Find}.

\item $P = \textsf{Complement}(P_1)$.
For this and all other top-level constructs, we will only consider the case that $P_v \neq P$.  Since $P_1$ was derived by the synthesis procedure, its goal annotation is $\phi_1 = (\absimg_{in} \setminus \absimg^+, \absimg_{in} \setminus \absimg^-)$. Note that $P_v$ is a complete sub-program of $P_1$, so $P_1$ is not consistent with $\absimg_{in}$. Let $\textsf{Complement}(P_1')$ be a completion of $P$, where $P_1'$ is a completion of $P_1$. By inductive hypothesis, $\bigsemantics{P_1'}(\absimg_{in}) \not\sim \phi_1$. By the semantics of \textsf{Complement},
\begin{equation}
\bigsemantics{\textsf{Complement}(P_1')}(\absimg_{in}) = \absimg_{in} \setminus \bigsemantics{P_1'}.
\end{equation}
Since $\bigsemantics{P_1'}(\absimg_{in}) \not\sim \phi_1$, it must be that either 
\begin{equation}
    \begin{split}
\absimg_{in} \setminus \absimg^+ &\not\subseteq \bigsemantics{P_1'}(\absimg_{in}) , \\
\text{or  } \bigsemantics{P_1'}(\absimg_{in}) &\not\subseteq \absimg_{in} \setminus \absimg^-.
    \end{split}
\end{equation}
 Without loss of generality, suppose the former, and notice that 
 \begin{equation}
 \absimg_{in} \setminus \bigsemantics{P_1'}(\absimg_{in}) = \bigsemantics{\textsf{Complement}(P_1')}(\absimg_{in}) \not\subseteq \absimg^+. 
 \end{equation}
 Then $\bigsemantics{\textsf{Complement}(P_1')}(\absimg_{in}) \not\sim \phi$. Therefore, for any completion $P'$ of $P$, $\bigsemantics{P'}(\absimg_{in}) \not\sim \phi$.

\item $P = \textsf{Union}(P_1, \ldots, P_n)$. Note that there is some $P_i$ such that $P_v$ is a complete sub-program of $P_i$. Then $P_i$ is not consistent with $\absimg_{in}$. Since $P$ was derived by \textsc{SynthesizeExtractor}, $P_i$ has goal annotation $\phi_i = (\emptyset, \absimg^+)$. Let $\textsf{Union}(P_1',\ldots,P_i',\ldots,P_n')$ be a completion of $P$, where for all $j$, $P_j'$ is a completion of $P_j$. By inductive hypothesis, $\bigsemantics{P_i'}(\absimg_{in}) \not\sim \phi_i$. By the semantics of \textsf{Union},
\begin{equation}
    \bigsemantics{\textsf{Union}(P_1',\ldots,P_i',\ldots,P_n')} = \bigcup_{k=1}^n \bigsemantics{P_k}(\absimg_{in}).
\end{equation}
 Since $\bigsemantics{P_i'}(\absimg_{in}) \not\sim \phi_i$, it must be that either $\emptyset \not\subseteq \bigsemantics{P_i'}(\absimg_{in})$ or $ \bigsemantics{P_i'}(\absimg_{in})\not\subseteq \absimg^+$. The former is not possible, so the latter must be true. Then there is some object $\imgobj \in \bigsemantics{P_i'}(\absimg_{in})$ such that $\imgobj \not\in \absimg^+$. Then it is also the case that
 \begin{equation}
 \imgobj \in \bigcup_{k=1}^n \bigsemantics{P_k}(\absimg_{in}) = \bigsemantics{\textsf{Union}(P_1',\ldots,P_i',\ldots,P_n')}(\absimg_{in},
 \end{equation} so
 \begin{equation}
 \bigsemantics{\textsf{Union}(P_1',\ldots,P_i',\ldots,P_n')} \not\subseteq \absimg^+.     
 \end{equation}
 Thus, $\textsf{Union}(P_1',\ldots,P_i',\ldots,P_n')$ is not consistent with $\absimg_{in}$. Therefore, for any completion $P'$ of $P$, $\bigsemantics{P'}(\absimg_{in}) \not\sim \phi$.

\item $P = \textsf{Intersect}(P_1, \ldots, P_n)$. Note that there is some $P_i$ such that $P_v$ is a complete sub-program of $P_i$. Then $P_i$ is not consistent with $\absimg_{in}$. Since $P_i$ was derived by the synthesis procedure, $P_i$ has goal annotation $\phi_i = (\absimg^-, \absimg_{in})$. 

Let $\textsf{Intersect}(P_1',\ldots,P_i',\ldots,P_n')$ be a completion of $P$, where for all $j$, $P_j'$ is a completion of $P_j$. By inductive hypothesis, $\bigsemantics{P_i'}(\absimg_{in}) \not\sim \phi_i$. By the semantics of \textsf{Intersect},
\begin{equation}
  \bigsemantics{\textsf{Intersect}(P_1',\ldots,P_i',\ldots,P_n')} = \bigcap_{k=1}^n \bigsemantics{P_k}(\absimg_{in}).  
\end{equation}
 Since $\bigsemantics{P_i'}(\absimg_{in}) \not\sim \phi_i$, it must be that either $\absimg^- \not\subseteq \bigsemantics{P_i'}(\absimg_{in})$ or $ \bigsemantics{P_i'}(\absimg_{in})\not\subseteq \absimg_{in}$. The latter is not possible, so the former must be true. Then there is some object $\imgobj \in \absimg^-$ such that $\imgobj \not\in \bigsemantics{P_i'}(\absimg_{in})$. Then it is also the case that 
 \begin{equation}
  \imgobj \not\in \bigcap_{k=1}^n \bigsemantics{P_k}(\absimg_{in}) = \bigsemantics{\textsf{Intersect}(P_1',\ldots,P_i',\ldots,P_n')}(\absimg_{in}),   
 \end{equation}
 so 
 \begin{equation}
 \absimg^- \not\subseteq \bigsemantics{\textsf{Intersect}(P_1',\ldots,P_i',\ldots,P_n')}. 
 \end{equation}
 Thus, $\textsf{Intersect}(P_1',\ldots,P_i',\ldots,P_n')$ is not consistent with $\absimg_{in}$. Therefore, for any completion $P'$ of $P$, $\bigsemantics{P'}(\absimg_{in}) \not\sim \phi$.
\end{itemize}
\end{proof}

\begin{theorem} \label{thm:consistency}
Let $P$ be a partial program derived by \textsc{SynthesizeExtractor} whose root node has goal annotation $(\absimg_{out}, \absimg_{out})$. If $P$ is not consistent with $\absimg_{in}$, then for any completion $P'$ of $P$, $\bigsemantics{P'}(\absimg_{in}) \not\equiv \absimg_{out}$.
\end{theorem}

\begin{proof}
Let $P$ be a partial program derived by the synthesis procedure whose root node has goal annotation $ (\absimg_{out}, \absimg_{out})$, and suppose $P$ is not consistent with $\absimg_{in}$. By Lemma \ref{lemma:consistency}, for any completion $P'$ of $P$, $\bigsemantics{P'}(\absimg_{in}) \not\sim (\absimg_{out}, \absimg_{out})$. Therefore, $\bigsemantics{P'}(\absimg_{in}) \not\equiv \absimg_{out}$. 
\end{proof}

\section{List of Benchmarks}


\begin{tabularx}{\textwidth}{|>{\hsize=.1\hsize\linewidth=\hsize}X|>{\hsize=1.6\hsize\linewidth=\hsize}X|>{\hsize=.3\hsize\linewidth=\hsize}X|
>{\hsize=.9\hsize\linewidth=\hsize}X|>{\hsize=.1\hsize\linewidth=\hsize}X|}
\hline
\textbf{id} & \textbf{Ground Truth Program} & \textbf{Dataset} & \textbf{Description} & \textbf{Size} \\
 \hline
 1 & \{\textsf{Intersect}(\textsf{Is}(\textsf{Smiling}), \textsf{Is}(\textsf{EyesOpen}))\newline $\rightarrow$\textsf{Brighten}\} & Wedding & Brighten all faces that are smiling and have eyes open. & 5 \\ 
 \hline
 2 & \{\textsf{Find}(\textsf{Is}(\textsf{FaceObject}), \newline \textsf{FaceObject}, \textsf{GetAbove}) \newline$\rightarrow$ \textsf{Brighten}\} & Wedding & Brighten all faces in back. & 5 \\ 
 \hline
 3 & \{\textsf{Union}(\textsf{Is}(\textsf{Face}(8)), \textsf{Is}(\textsf{Face}(34))) $\rightarrow$ \textsf{Crop}\} & Wedding & Crop image to feature just faces of bride and groom. & 7 \\ 
 \hline
 4 & \{\textsf{Intersect}(\textsf{Is}(\textsf{FaceObject}), \newline \textsf{Complement}(\textsf{Is}(\textsf{Face}(8)))) $\rightarrow$ \textsf{Blur}\} & Wedding & Blur all faces except the bride's face. & 7 \\
 \hline
  5 & \{\textsf{Find}(\textsf{Find}(\textsf{Is}(\textsf{FaceObject}), \newline \textsf{FaceObject}, \textsf{GetRight}), \textsf{FaceObject}, \newline \textsf{GetRight}) $\rightarrow$ \textsf{Brighten}\} & Wedding & Brighten all faces except the leftmost two faces. & 8 \\ 
 \hline
 6 & \{\textsf{Intersect}(\textsf{Is}(\textsf{Face}), \newline \textsf{Complement}(\textsf{Intersect}(\textsf{Is}(\textsf{Smiling}), \textsf{Is}(\textsf{EyesOpen}))) $\rightarrow$ \textsf{Blur} \} & Wedding & Blur all faces that are not smiling and do not have their eyes open. & 9 \\
 \hline
 7 & \{\textsf{Intersect}(\textsf{Is}(\textsf{Smiling}), \textsf{Is}(\textsf{EyesOpen}),
 \newline
 \textsf{Complement}(\textsf{Is}(\textsf{Face}(34))))$\rightarrow$ \textsf{Blur}\} & Wedding & Crop image to feature all faces that are smiling and have eyes open, except the groom's face. & 9 \\
 \hline
 8 & \{\textsf{Union}(\textsf{Is}(\textsf{Face}(8)), \newline  \textsf{Intersect}(\textsf{Is}(\textsf{Smiling}), \textsf{Is}(\textsf{EyesOpen}))) $\rightarrow$ \textsf{Blur}\}& Wedding & Crop image to feature the bride's face, plus faces that are smiling and have their eyes open. & 9 \\
 \hline 
 9 & \{\textsf{Intersect}(\textsf{Complement}(\textsf{Is}(\textsf{Smiling})), \textsf{Find}(\textsf{Is}(\textsf{FaceObject}), \newline
 \textsf{FaceObject}, \textsf{GetAbove})) $\rightarrow$ \textsf{Blur}\} & Wedding & Blur all faces in the back that are not smiling. & 9 \\ 
 \hline
 10 & \{\textsf{Union}(\textsf{Intersect}(\textsf{Is}(\textsf{FaceObject}), \textsf{Complement}(\textsf{Is}(\textsf{Smiling}))), \textsf{Is}(\textsf{BelowAge}(18))) $\rightarrow$ \textsf{Blur}\} & Wedding & Blur all faces that are not smiling or are under 18. & 10 \\ 
 \hline 
 11 & \{\textsf{Union}(\textsf{Find}(\textsf{Is}(\textsf{Face}(8)), \textsf{FaceObject}, \textsf{GetRight}), \textsf{Is}(\textsf{Face}(8))) $\rightarrow$ \textsf{Crop}\} & Wedding & Crop image to feature just the bride's face and the face directly to her right. & 10 \\ 
 \hline
 12 & \{\textsf{Union}(\textsf{Is}(\textsf{Face}(8)), \newline \textsf{Find}(\textsf{Is}(\textsf{Face}(8)), \textsf{Face}(34), \textsf{GetAbove})) $\rightarrow$ \textsf{Crop}\} & Wedding & Crop image to feature just the bride and the groom when he is behind her. & 11 \\ 
 \hline
  13 & \{\textsf{Intersect}(\textsf{Find}(\textsf{Is}(\textsf{FaceObject}), \textsf{FaceObject}, \textsf{GetRight}), \newline \textsf{Find}(\textsf{Is}(\textsf{FaceObject}), \textsf{FaceObject}, \textsf{GetLeft})) $\rightarrow$ \textsf{Brighten}\} & Wedding & Brighten all faces except leftmost and rightmost face. & 11 \\ 
 \hline
 14 & \{\textsf{Find}(\textsf{Union}(\textsf{Is}(\textsf{Face}(34)), \textsf{Is}(\textsf{Smiling}), \textsf{Is}(\textsf{EyesOpen})), \newline \textsf{Object}(\textsf{Person}), \textsf{GetBelow}) $\rightarrow$ \textsf{Sharpen}\} & Wedding & Sharpen the groom, and all smiling people and people with their eyes open. & 12 \\
 \hline
 \end{tabularx}

 \begin{tabularx}{\textwidth}{|>{\hsize=.1\hsize\linewidth=\hsize}X|>{\hsize=1.6\hsize\linewidth=\hsize}X|>{\hsize=.3\hsize\linewidth=\hsize}X|
>{\hsize=.9\hsize\linewidth=\hsize}X|>{\hsize=.1\hsize\linewidth=\hsize}X|}
\hline
 15 & \{\textsf{Intersect}(\textsf{Find}(\textsf{Is}(\textsf{FaceObject}), \newline  \textsf{Face}(8), \textsf{GetRight}), \newline  \textsf{Find}(\textsf{Is}(\textsf{FaceObject}), \textsf{Face}(8), \newline \textsf{GetLeft})) $\rightarrow$ \textsf{Crop}\} & Wedding & Crop image to feature just bride when someone is to her left and right. & 13 \\ 
 \hline
 16 & \{\textsf{Union}(\textsf{Find}(\textsf{Is}(\textsf{Face}(8)), \textsf{FaceObject}, \textsf{GetRight}), \textsf{Find}(\textsf{Is}(\textsf{Face}(8)), \newline \textsf{FaceObject}, \textsf{GetLeft}), \textsf{Is}(\textsf{Face}(8))) \newline $\rightarrow$ \textsf{Crop}\} & Wedding & Crop image to feature just the bride and the people to her left and right. & 16 \\
 \hline
 17 & \{\textsf{Union}(\textsf{Is}(\textsf{Price}), \textsf{Is}(\textsf{PhoneNumber})) $\rightarrow$ \textsf{Blackout}\} & Receipts & Blackout all prices and phone numbers. & 5 \\ 
 \hline 
  18 & \{\textsf{Find}(\textsf{Is}(\textsf{Price}), \textsf{TextObject}, \textsf{GetLeft}) $\rightarrow$ \textsf{Brighten}\} & Receipts & Brighten text to the left of a price. & 5 \\ 
 \hline 
 19 & \{\textsf{Intersect}(\textsf{Is}(\textsf{Text}), \newline  \textsf{Complement}(\textsf{Is}(\textsf{Price}))) $\rightarrow$ \text{Blackout} \} & Receipts & Blackout all text that is not a price. & 6 \\ 
 \hline 
 20 & \{\textsf{Find}(\textsf{Is}(\textsf{Word}(\texttt{"total"})), \newline \textsf{Price}, \textsf{GetRight}) $\rightarrow$ \textsf{Brighten}\} & Receipts & Brighten all prices to the right of the word "total." & 6 \\ 
 \hline 
 21 & \{\textsf{Find}(\textsf{Is}(\textsf{Word}(\texttt{"total"})), \textsf{TextObject}, \textsf{GetRight}) $\rightarrow$ \textsf{Brighten} & Receipts & Brighten text to the right of the word "total." & 6 \\ 
 \hline 
 22 & \{\textsf{Find}(\textsf{Is}(\textsf{Word}(\texttt{"tax"})), \newline \textsf{TextObject}, \textsf{GetAbove}) $\rightarrow$ \textsf{Blackout}\} & Receipts & Blackout all text above the word "tax." & 6 \\
 \hline 
  23 & \{\textsf{Find}(\textsf{Find}(\textsf{Is}(\textsf{TextObject}), \textsf{TextObject}, \textsf{GetLeft}), \textsf{Text}, \textsf{GetLeft}) $\rightarrow$ \textsf{Brighten}\} & Receipts & Brighten all text except rightmost two columns. & 8 \\ 
 \hline 
 24 & \{\textsf{Intersect}(\textsf{Is}(\textsf{Text}), \newline  \textsf{Complement}(\textsf{Union}(\textsf{Is}(\textsf{Price}), \textsf{Is}(\textsf{PhoneNumber})))) $\rightarrow$ \textsf{Blackout}\} & Receipts & Blackout all text that is not a price or a phone number. & 9 \\ 
 \hline 
  25 & \{\textsf{Find}(\textsf{Find}(\textsf{Is}(\textsf{Word}(\texttt{"total"})), \textsf{Price}, \textsf{GetRight}), \textsf{IsPrice}, \textsf{GetAbove}) $\rightarrow$ \textsf{Brighten}\} & Receipts & Brighten the price that is above the total price. & 9 \\ 
  \hline
 26 & \{\textsf{Complement}(\textsf{Find}(\textsf{Find}(\textsf{Is}(\textsf{TextObject}), \textsf{Find}(\textsf{TextObject}), \textsf{GetAbove}), \textsf{TextObject}, \textsf{GetAbove})) $\rightarrow$ \textsf{Blackout} \} & Receipts & Blackout bottom two rows of text. & 10 \\ 
 \hline 
 27 & \{\textsf{Intersect}(\textsf{Is}(\textsf{TextObject}), \newline  \textsf{Complement}(\textsf{Union}(\textsf{Is}(\textsf{Word}(\texttt{"total"})), \textsf{Is}(\textsf{Price})))) $\rightarrow$ \textsf{Blackout}\} & Receipts & Blackout all text except prices and the word "total." & 10 \\ 
 \hline 
 28 & \{\textsf{Intersect}(\textsf{Is}(\textsf{Price}), \newline  \textsf{Complement}(\textsf{Find}(\textsf{Is}(\textsf{Word}(\texttt{"total"})), \textsf{TextObject}, \textsf{GetRight}))) $\rightarrow$ \textsf{Blackout}\} & Receipts & Blackout all prices that are not the total price. & 10 \\ 
 \hline 
 29 & \{\textsf{Union}(\textsf{Find}(\textsf{Is}(\textsf{Word}(\texttt{"total"})), \newline  \textsf{TextObject}, \textsf{GetRight}), \newline  \textsf{Find}(\textsf{Is}(\textsf{Word}(\texttt{"subtotal"})), \newline \textsf{TextObject}, \textsf{GetRight})) $\rightarrow$ \textsf{Blackout}\} & Receipts & Blackout all prices that are not the total price or subtotal price. & 13 \\ 
 \hline 
 30 & \{\textsf{Complement}(\textsf{Is}(\textsf{Object}(\texttt{car}))) $\rightarrow$ \textsf{Blur}\} & Objects & Blur all objects except cars. & 4 \\ 
 \hline
 31 & \{\textsf{Filter}(\textsf{Is}(\textsf{Object}(\texttt{car})), \textsf{FaceObject}) $\rightarrow$ \textsf{Blur}\} & Objects & Blur all faces in cars. & 5 \\ 
 \hline 
 32 & \{\textsf{Filter}(\textsf{Is}(\textsf{Object}(\texttt{car})), \textsf{TextObject}) $\rightarrow$ \textsf{Blur}\} & Objects & Blur all text on cars. & 5 \\
 \hline 
\end{tabularx}

\begin{tabularx}{\textwidth}{|>{\hsize=.1\hsize\linewidth=\hsize}X|>{\hsize=1.6\hsize\linewidth=\hsize}X|>{\hsize=.3\hsize\linewidth=\hsize}X|
>{\hsize=.9\hsize\linewidth=\hsize}X|>{\hsize=.1\hsize\linewidth=\hsize}X|}
\hline
 33 & \{\textsf{Find}(\textsf{Is}(\textsf{TextObject}), \textsf{Object}(\texttt{car}), \newline \textsf{GetParents}) $\rightarrow$ \textsf{Blur}\} & Objects & Blur all cars with text on them. & 6 \\ 
 \hline 
 34 & \{\textsf{Union}(\textsf{Is}(\textsf{Object}(\texttt{cat})), \textsf{Is}(\textsf{FaceObject})) \newline $\rightarrow$  \textsf{Brighten}\} & Objects & Brighten all faces and all cats. & 6 \\ 
 \hline 
 35 & \textsf{Union}(\textsf{Is}(\textsf{Object}(\texttt{cat})), \textsf{Is}(\textsf{EyesOpen})) $\rightarrow$ \textsf{Brighten}\} & Objects & Brighten all faces with eyes open and all cats. & 6 \\ 
 \hline
  36 & \{\textsf{Find}(\textsf{Is}(\textsf{Object}(\texttt{guitar})), \textsf{FaceObject}, \textsf{GetAbove}) $\rightarrow$ \textsf{Sharpen}\} & Objects & Sharpen faces of people playing guitar. & 6 \\ 
 \hline 
 37 & \{\textsf{Find}(\textsf{Is}(\textsf{Word}(\texttt{319})), \textsf{Object}(\texttt{car}), \newline \textsf{GetParents}) $\rightarrow$ \textsf{Blur}\} & Objects & Blur car with number 319. & 7 \\ 
 \hline 
 38 & \{\textsf{Union}(\textsf{Is}(\textsf{Object}(\texttt{car})), \textsf{Is}(\textsf{Object}(\texttt{bicycle}))) $\rightarrow$ \textsf{Brighten}\} & Objects & Brighten all cars and bicycles. & 7 \\ 
 \hline 
  39 & \{\textsf{Find}(\textsf{Is}(\textsf{Object}(\texttt{person})), \newline  \textsf{Object}(\texttt{bicycle}), \textsf{GetBelow}) $\rightarrow$ \textsf{Brighten}\} & Objects & Brighten all bicycles that are being ridden. & 7 \\ 
 \hline 
  40 & \{\textsf{Find}(\textsf{Is}(\textsf{Object}(\texttt{bicycle})), \newline  \textsf{BelowAge}(\texttt{18}), \textsf{GetAbove}) $\rightarrow$ \textsf{Blur}\} & Objects & Blur the faces of children riding bicycles. & 7 \\ 
 \hline 
 41 & \{\textsf{Complement}(\textsf{Union}(\textsf{Is}(\textsf{Object}(\texttt{car})), \textsf{Is}(\textsf{Object}(\texttt{bicycle})))) $\rightarrow$ \textsf{Blackout}\} & Objects & Blackout all objects except cars and bicycles. & 8 \\ 
 \hline 
  42 & \{\textsf{Intersect}(\textsf{Is}(\textsf{TextObject}), \newline  \textsf{Complement}(\textsf{Filter}(\textsf{Is}(\textsf{Object}(\texttt{car})), \newline  \textsf{TextObject}))) $\rightarrow$ \textsf{Blackout}\} & Objects & Blackout all text not on a car. & 9 \\
  \hline 
 43 & \{\textsf{Union}(\textsf{Is}(\textsf{Object}(\texttt{bicycle})), \textsf{Is}(\textsf{Object}(\texttt{car})), \textsf{Is}(\textsf{Object}(\texttt{person}))) $\rightarrow$ \textsf{Brighten}\} & Objects & Brighten all bicycles, cars, and people. & 10 \\ 
 \hline 
 44 & \{\textsf{Intersect}(\textsf{Is}(\textsf{FaceObject}), \newline  \textsf{Complement}(\textsf{Find}(\textsf{Is}(\textsf{Object}(\texttt{bicycle})), \textsf{FaceObject}, \textsf{GetAbove}))) $\rightarrow$ \textsf{Blur}\} & Objects & Blur faces of people not riding bicycles. & 10 \\ 
 \hline 
 45 & \{\textsf{Union}(\textsf{Is}(\textsf{Object}(\texttt{guitar})), \textsf{Find}(\textsf{Is}(\textsf{Object}(\texttt{guitar})), \textsf{FaceObject}, \textsf{GetAbove})) $\rightarrow \textsf{Brighten}$\} & Objects & Brighten all guitars and people playing guitar. & 10 \\ 
 \hline 
 46 & \{\textsf{Intersect}(\textsf{Is}(\textsf{Face}), \newline  \textsf{Complement}(\textsf{Find}(\textsf{Is}(\textsf{Object}(\texttt{guitar})), \textsf{FaceObject}, \textsf{GetAbove}))) $\rightarrow$ \textsf{Blur}\} & Objects & Blur faces of people not playing guitar. & 10 \\ 
 \hline
 47 & \{\textsf{Intersect}(\textsf{Is}(\textsf{Object}(\texttt{bicycle})), \newline  \textsf{Complement}(\textsf{Find}(\textsf{Is}(\textsf{Object}(\texttt{person})), \textsf{Object}(\texttt{bicycle}), \textsf{GetBelow}))) $\rightarrow$ \textsf{Sharpen}\} & Objects & Sharpen bicycles that are not being ridden. & 12 \\ 
 \hline 
 48 & \{\textsf{Intersect}(\textsf{Is}(\textsf{Object}(\texttt{bicycle})), \textsf{Complement}(\textsf{Find}(\textsf{Is}(\textsf{BelowAge}(\texttt{18})), \textsf{Object}(\texttt{bicycle}), \textsf{GetBelow}))) $\rightarrow$ \textsf{Sharpen}\} & Objects & Sharpen all bicycles that are not ridden by a child. & 12 \\ 
 \hline 
 49 & \{\textsf{Intersect}(\textsf{Is}(\textsf{Object}(\texttt{cat})), 
 \newline \textsf{Complement}(\textsf{Find}(\textsf{Is}(\textsf{Object}(\texttt{cat})), \newline  \textsf{Object}(\texttt{cat}), \textsf{GetBelow}))) $\rightarrow$ \textsf{Crop} \} & Objects & Crop image to feature just topmost cat. & 12 \\ 
 \hline 
 50 & \{\textsf{Intersect}(\textsf{Find}(\textsf{Is}(\textsf{Object}(\texttt{cat})), \newline \textsf{Object}(\texttt{cat}), \textsf{GetRight}), \newline \textsf{Find}(\textsf{Is}(\textsf{Object}(\texttt{cat})), \textsf{Object}(\texttt{cat}), \newline \textsf{GetLeft})) $\rightarrow$ \textsf{Brighten}\} & Objects & Brighten cats that are between two other cats. & 15 \\ 
 \hline 
 \end{tabularx}
 
 \section{List of Neural Attributes}

\begin{itemize}
    \item \textsf{FaceObject}. Objects identified by Amazon Rekognition's face recognition model.
    \item \textsf{Smiling}. Faces identified as smiling by Rekognition's face recognition model.
    \item \textsf{EyesOpen}. Faces identified as having open eyes by Rekognition's face recognition model.
    \item \textsf{MouthOpen}. Face identified as having open mouths by Rekognition's face recognition model.
    \item \textsf{BelowAge}$(N)$. Rekognition's face recognition model returns an upper and lower bound on a face's age. This predicate matches a face if the upper bound on its age is less than $N$.
    \item \textsf{Face}$(N)$. Rekognition assigns each unique face an identified. This predicate matches a face if its identified is equal to $N$.
    \item \textsf{TextObject}. Objects identified by Amazon Rekognition's text recognition model. 
    \item \textsf{Word}$(W)$. Rekognition's text recognition model returns the text value of each identified text object. This predicate matches a text object if its text value equals $W$.
    \item \textsf{PhoneNumber}. Text objects whose value matches the format of a North American phone number. 
    \item \textsf{Price}. Text objects whose value matches the format of a price.
    \item \textsf{Object}$(O)$. Objects identified by Rekognition's object recognition model. This model identifies 238 unique types of objects. This predicate matches an object whose type is equal to $O$.

\end{itemize}

\end{document}